\documentclass[12pt,reqno]{article}

\usepackage{amsmath}
\usepackage{amssymb}
\usepackage{amsthm}
\usepackage{bm}
\usepackage[bf,footnotesize]{caption}
\usepackage{color}
\usepackage{comment}
\usepackage[shortlabels]{enumitem}
\usepackage{eucal}
\usepackage{fullpage}
\usepackage{graphicx}
\usepackage{lineno}
\usepackage{mathpazo}
\usepackage{mathrsfs}
\usepackage{mathtools}
\usepackage[numbers,sort&compress,square]{natbib}
\usepackage{setspace}
\usepackage{subfig}
\usepackage{times}
\usepackage[all,cmtip]{xy}

\newcommand*\patchAmsMathEnvironmentForLineno[1]{%
	\expandafter\let\csname old#1\expandafter\endcsname\csname #1\endcsname
	\expandafter\let\csname oldend#1\expandafter\endcsname\csname end#1\endcsname
	\renewenvironment{#1}%
	{\linenomath\csname old#1\endcsname}%
	{\csname oldend#1\endcsname\endlinenomath}}% 
	\newcommand*\patchBothAmsMathEnvironmentsForLineno[1]{%
	\patchAmsMathEnvironmentForLineno{#1}%
	\patchAmsMathEnvironmentForLineno{#1*}}%
	\AtBeginDocument{%
	\patchBothAmsMathEnvironmentsForLineno{equation}%
	\patchBothAmsMathEnvironmentsForLineno{align}%
	\patchBothAmsMathEnvironmentsForLineno{flalign}%
	\patchBothAmsMathEnvironmentsForLineno{alignat}%
	\patchBothAmsMathEnvironmentsForLineno{gather}%
	\patchBothAmsMathEnvironmentsForLineno{multline}%
}

\newcommand{\eq}[1]{\textbf{Eq.~\ref{eq:#1}}}
\newcommand{\fig}[1]{\textbf{Fig.~\ref{fig:#1}}}
\newcommand{\prop}[1]{\textbf{Proposition~\ref{prop:#1}}}
\newcommand{\thm}[1]{\textbf{Theorem~\ref{thm:#1}}}
\newcommand{\sect}[1]{Section~\ref{sec:#1}}

\newtheorem{corollary}{Corollary}

\newtheorem{proposition}{Proposition}
\newtheorem{theorem}{Theorem}
\newtheorem*{equalizertheorem}{Theorem~\ref{thm:equalizer}}
\newtheorem*{propendpoints}{Proposition~\ref{prop:endpoints}}

\theoremstyle{definition}
\newtheorem{remark}{Remark}

\title{\begin{center} \bfseries \singlespacing ~\\[-0.8cm]
Payoff landscapes and the robustness of selfish optimization in iterated games
\end{center}}
\author{\parbox[c]{16cm}{\onehalfspacing \centering ~\\[-0.5cm] Arjun Mirani$^{1,2}$ and Alex McAvoy$^{3,4}$\\ \quad\\ \footnotesize
$^{1}$Harvard College, Cambridge, MA, USA \\
$^{2}$Department of Applied Mathematics and Theoretical Physics, University of Cambridge, Cambridge, UK \\
$^{3}$Department of Mathematics, University of Pennsylvania, Philadelphia, PA, USA \\
$^{4}$Center for Mathematical Biology, University of Pennsylvania, Philadelphia, PA, USA \\[0.1cm]}
\date{}
}

\begin{document}

\allowdisplaybreaks

\maketitle

{\fontfamily{put}\selectfont}

\begin{abstract}
In iterated games, a player can unilaterally exert influence over the outcome through a careful choice of strategy. A powerful class of such ``payoff control'' strategies was discovered by \citet{pressdyson}. Their so-called ``zero-determinant'' (ZD) strategies allow a player to unilaterally enforce a linear relationship between both players' payoffs. It was subsequently shown by \citet{chenzinger} that when the slope of this linear relationship is positive, ZD strategies are robustly effective against a selfishly optimizing co-player, in that all adapting paths of the selfish player lead to the maximal payoffs for both players (at least when there are certain restrictions on the game parameters). In this paper, we investigate the efficacy of selfish learning against a fixed player in more general settings, for both ZD and non-ZD strategies. We first prove that in any symmetric $2\times 2$ game, the selfish player's final strategy must be of a certain form and cannot be fully stochastic. We then show that there are prisoner's dilemma interactions for which selfish optimization does not always lead to maximal payoffs against fixed ZD strategies with positive slope. We give examples of selfish adapting paths that lead to locally but not globally optimal payoffs, undermining the robustness of payoff control strategies. For non-ZD strategies, these pathologies arise regardless of the original restrictions on the game parameters. Our results illuminate the difficulty of implementing robust payoff control and selfish optimization, even in the simplest context of playing against a fixed strategy.
\end{abstract}

\section{Introduction}
Simple mathematical models of interactions abound in the literature on theoretical biology and the social sciences. Iterated games, which have been used extensively to study reciprocation of altruistic behaviors, comprise one such class of models. Although iterated games abstract away many details of realistic encounters, they have been enormously useful in establishing the theoretical foundations for phenomena such as the widespread prevalence of prosocial behaviors \citep{axelrodbook,rapoport,maynard,sigmundbook}. And it is likely that no iterated game has received as much attention as the iterated prisoner's dilemma \citep{nowak:fiverules,nowak:book,doebeli:ecology:2005,lorens:PNAS:2005,segbroeck:PRL:2012}.

Every round of the iterated prisoner's dilemma (IPD) is a ``one-shot'' prisoner's dilemma game, set up as follows. Each player can either cooperate ($C$) or defect ($D$), with their actions chosen simultaneously. For outcomes $\left(C,C\right)$, $\left(C,D\right)$, $\left(D,C\right)$, and $\left(D,D\right)$, the payoffs to the two players are $\left(R,R\right)$, $\left(S,T\right)$, $\left(T,S\right)$, $\left(P,P\right)$, respectively. Since prisoner's dilemma games are characterized by $T>R>P>S$, in the absence of knowing the co-player's action, each player would rationally choose to defect, making mutual defection the unique Nash equilibrium for the one-shot game \citep{nash}. However, mutual cooperation yields a higher payoff for each player than mutual defection does. Recognizing this, players who repeatedly interact can choose to condition their future actions (perhaps probabilistically) on the outcome of earlier rounds, as humans do in social dilemmas \citep{rand:CogSci:2013}. Such conditional behavior enables direct reciprocity, an important mechanism for the evolution of cooperation: $X$ will cooperate with $Y$ if $Y$ has a history of cooperating with $X$, and vice versa \citep{baek:Nature:2016}.

Since conditional strategies require players to remember previous rounds, an important consideration is the memory length of a player. A memory-$n$ strategy conditions the next move on the previous $n$ rounds \citep{hilbe:PNAS:2017}. A player with infinite memory uses the entire history of play to determine his or her next action. Here, we restrict our attention to memory-one strategies, which enable a rich variety of behaviors while remaining analytically tractable and easy to implement. This is an extremely common assumption in the literature, and indeed some of the most successful strategies, such as ``tit-for-tat,'' ``generous tit-for-tat,'' and ``win-stay, lose-shift,'' are memory-one strategies \citep{nowak:WSLS,rand:JTB:2009}.

After several decades of thorough analysis, interest in the IPD was revitalized by Press and Dyson's discovery of ``zero-determinant'' (ZD) strategies in 2012 \citep{pressdyson}. This powerful class of memory-one strategies allows a player to unilaterally impose linear relationships between both players' payoffs, or even set the co-player's payoff to a fixed value regardless of the latter's strategy, thereby enabling behaviors like extortion, generosity, and fairness \citep{hilbe:adaptive,pressdyson,stewart:PNAS:2013}. ZD strategies are a method of ``payoff control'' in that they allow a player to restrict the set of feasible payoffs of the game. More general approaches to payoff control, involving linear inequalities among the payoffs, include what are known as ``partner'' and ``rival'' strategies \citep{hilbe:GEB:2015,payoffcontrol,partnersrivals}. By appropriately restricting the feasible payoffs, such strategies can, in principle, align a self-concerned co-player's optimal play with a desired behavior.

But this observation raises an important question: is it even possible for a selfish player, who aims to optimize his or her own payoff only, to fully accomplish this goal? \citet{pressdyson} conjectured, based on numerical simulations, that a selfish adapting player always learns to cooperate against a ZD player who imposes a positive-slope linear payoff constraint that unfairly benefits the fixed player (an ``extortionate'' strategy). \citet{chenzinger} formally proved this conjecture (in fact, a stronger version) for IPDs satisfying $2P<S+T<2R$. In particular, all adapting paths of the selfish player lead to repeated cooperation, maximizing both players' payoffs. Therefore, by using an appropriate ZD strategy to align the incentives of the two players \emph{in principle}, an individual can elicit cooperation from a purely self-motivated co-player \emph{in practice}.

In this paper, we investigate the robustness of payoff control strategies in achieving the desired outcome against a selfish learner who optimizes through local hill-climbing. We build on \citep{chenzinger} by considering whether robustness holds under different conditions on the game parameters, as well as for non-ZD strategies. First, we analytically establish constraints on the optima of the payoff landscape over which the selfish learner optimizes. We use these constraints to prove that the selfish player's learned strategy must be of a certain form and in general cannot be a fully mixed strategy. These results apply not just to the IPD, but to all iterated, symmetric, two-player, two-action games. Next, using gradient-based optimization against an opponent with a ZD strategy of positive slope, we demonstrate that the adapting player often ends up at locally, but not globally, optimal strategies in IPDs for which $S+T>2R$ or $S+T<2P$. When the fixed player uses a general memory-one strategy, these pathologies can arise even when $2P<S+T<2R$. Finally, we conclude with a discussion of the effects of noise, both in the optimization procedure itself and in the implementation of strategies, which we find to mitigate some of the inefficiencies of gradient ascent.

\section{Model and summary of previous results}
We consider two players, $X$ and $Y$, engaged in an iterated game. In each round, the players choose their respective actions $x,y\in\left\{C,D\right\}$. The payoffs in each round are specified by a vector $\left(R,S,T,P\right)\in\mathbb{R}^{4}$, corresponding to the joint actions $\left(CC,CD,DC,DD\right)$ from the focal player's perspective (e.g. if $X$ plays $C$ and $Y$ plays $D$, then $X$ gets $S$ and $Y$ gets $T$). Different payoff vectors define different games. The class of prisoner's dilemma interactions satisfies $T>R>P>S$, which ensures that defection is the dominant strategy. An additional condition $2R >S+T$ ensures that mutual cooperation is the socially optimal outcome. If instead we have $S+T>2R$, the socially optimal behavior is anti-coordinated alternation between cooperation and defection, with one player using $C$ in even rounds only and the other using $C$ in odd rounds only.

When the game is iterated, $X$ and $Y$ choose their next action based on the actions of the previous round. Each player's choice is governed by a memory-one strategy, which specifies probabilities $p_{xy}\in\left[0,1\right]$ for cooperating in the next round given each possible action $x$ of $X$ and $y$ of $Y$ in the previous round, where $x,y\in\left\{C,D\right\}$. Therefore, a memory-one strategy is specified by a $4$-tuple of probabilities, $\mathbf{p}=\left(p_{C C},p_{C D},p_{D C},p_{D D}\right)\in\left[0,1\right]^{4}$, along with the probability of cooperating in the initial round, $p_{0}\in\left[0,1\right]$ (where there is no history on which to condition). A component $p_{xy}$ corresponds to a ``pure'' action if it is $0$ (defection) or $1$ (cooperation); otherwise, it is ``mixed.''

The initial actions define a distribution over action pairs, $\left(CC,CD,DC,DD\right)$,
\begin{align}
\nu_{0}\left(p_{0},q_{0}\right) &= \left(p_{0}q_{0},p_{0}\left(1-q_{0}\right) ,\left(1-p_{0}\right) q_{0},\left(1-p_{0}\right)\left(1-q_{0}\right)\right) .
\end{align}
Given conditional strategies $\mathbf{p}=\left(p_{CC},p_{CD},p_{DC},p_{DD}\right)$ for $X$ and $\mathbf{q}=\left(q_{CC},q_{CD},q_{DC},q_{DD}\right)$ for $Y$, the iterated game defines a Markov chain on $\left(CC,CD,DC,DD\right)$ whose transition matrix is
\begin{align}
M &= 
\bordermatrix{%
& CC & CD & DC & DD \cr
CC &\ p_{CC} q_{CC} &\ p_{CC}\left(1-q_{CC}\right) &\ \left(1-p_{CC}\right) q_{CC} &\ \left(1-p_{CC}\right)\left(1-q_{CC}\right) \cr
CD &\ p_{CD} q_{DC} &\ p_{CD}\left(1-q_{DC}\right) &\ \left(1-p_{CD}\right) q_{DC} &\ \left(1-p_{CD}\right)\left(1-q_{DC}\right) \cr
DC &\ p_{DC} q_{CD} &\ p_{DC}\left(1-q_{CD}\right) &\ \left(1-p_{DC}\right) q_{CD} &\ \left(1-p_{DC}\right)\left(1-q_{CD}\right) \cr
DD &\ p_{DD} q_{DD} &\ p_{DD}\left(1-q_{DD}\right) &\ \left(1-p_{DD}\right) q_{DD} &\ \left(1-p_{DD}\right)\left(1-q_{DD}\right) \cr
}\ . \label{eq:markov_matrix}
\end{align}
If $\nu_{t}$ is the distribution over actions at time $t$, then we have $\nu_{t+1}=\nu_{t}M$ for all $t\geqslant 0$.

The expected payoffs to $X$ and $Y$ in round $t$ are then $\pi_{X,t}=\left<\nu_{t},\left(R,S,T,P\right)\right>$ and $\pi_{Y,t}=\left<\nu_{t},\left(R,T,S,P\right)\right>$, respectively, where $\left<\cdot ,\cdot\right>$ is the standard inner (dot) product on $\mathbb{R}^{4}$. In general, an iterated game has a discounting factor $\lambda <1$ specifying the probability that game proceeds one round further. The expected payoffs to the players, $\pi_{X}$ and $\pi_{Y}$, are calculated by discounting the one-shot payoffs in round $t$ by a factor of $\lambda^{t}$, summing these discounted payoffs in the limit $t\rightarrow\infty$, and dividing the results by the expected number of rounds, $1/\left(1-\lambda\right)$. In other words, $\pi_{X}=\left(1-\lambda\right) \sum_{t=0}^{\infty} \lambda^{t} \pi_{X,t}$ and $\pi_{Y}=\left(1-\lambda\right) \sum_{t=0}^{\infty} \lambda^{t} \pi_{Y,t}$. Letting $I$ denote the identity matrix, and using the recurrence for $\nu_{t}$ mentioned previously, one obtains the following expressions:
\begin{subequations}
\begin{align}
\pi_{X} &= \left< \left(1-\lambda\right)\nu_{0} \left(I-\lambda M\right)^{-1} , \left(R,S,T,P\right) \right> ; \\
\pi_{Y} &= \left< \left(1-\lambda\right)\nu_{0} \left(I- \lambda M\right)^{-1} , \left(R,T,S,P\right) \right> .
\end{align}
\end{subequations}

Here, we take the limit $\lambda \to 1^{-}$ so that the game is iterated infinitely. When the Markov chain defined by \eq{markov_matrix} has a unique limiting distribution, $\nu$, we can then write $\pi_{X}=\left<\nu ,\left(R,S,T,P\right)\right>$ and $\pi_{Y}=\left<\nu ,\left(R,T,S,P\right)\right>$. Such is the case when $\mathbf{p},\mathbf{q}\in\left(0,1\right)^{4}$, and then $\nu$ is the unique probability vector satisfying $\nu M=\nu$. More generally, $\nu$ is given by the rows of the Ces\`{a}ro limit
\begin{align}
M^{\ast} &= \lim_{n\rightarrow\infty} \frac{1}{n} \sum_{k=1}^{n} M^{k} .
\end{align}
When the chain is aperiodic, $M^{\ast}$ agrees with $\lim_{n\rightarrow\infty}M^{n}$. For the vast majority of strategy pairs $\mathbf{p}$ and $\mathbf{q}$, the stationary distribution is unique; the Markov chain is either irreducible or, if it has transient states, it has a unique closed communicating class. In this case, $M^{\ast}$ is a rank-one matrix with all its rows proportional to $\nu$. Multiple stationary distributions occur only for certain combinations of $0$s and $1$s in $\mathbf{p}$ and $\mathbf{q}$, in which case the long-run payoffs depend on the initial state. A more detailed discussion of edge cases can be found in \citep{chenzinger}. It is useful to note that when either $\mathbf{p}$ or $\mathbf{q}$ is in $\left(0,1\right)^{4}$, $\nu$ is well-defined for all values of the other strategy except $\left(1,1,0,0\right)$ (``repeat'').

\subsection{Payoff control and zero-determinant strategies}
The payoffs in an iterated interaction clearly depend on both players' strategies. A natural question is then the following: to what extent can one player, say $X$, constrain the other player's payoff through a careful choice of strategy, $\mathbf{p}$? A striking method of payoff control was discovered by \citet{pressdyson}, who showed that a player can unilaterally impose a linear relationship between the payoffs, regardless of the co-player's strategy. Such a ``zero-determinant'' (ZD) strategy $\mathbf{p}$ enforces the relationship $\pi_{X}-\kappa =\chi\left(\pi_{Y}-\kappa\right)$ whenever there exists $\phi\in\mathbb{R}$ with
\begin{subequations}\label{eq:ZD_formula}
\begin{align}
p_{CC} &= 1-\phi\left(\chi -1\right)\left(R-\kappa\right) ; \\
p_{CD} &= 1-\phi\left(\kappa -S+\chi\left(T-\kappa\right)\right) ; \\
p_{DC} &= \phi\left(T-\kappa + \chi\left(\kappa -S\right)\right) ; \\
p_{DD} &= \phi\left(\chi -1\right)\left(\kappa -P\right) .
\end{align}
\end{subequations}
For given $\kappa$ and $\chi$, there might be many (or no) strategies $\mathbf{p}$ that enforce $\pi_{X}-\kappa =\chi\left(\pi_{Y}-\kappa\right)$, corresponding to a range of allowed values of $\phi$ that ensure $\mathbf{p}\in\left[0,1\right]^{4}$. For instance, in an IPD, there are feasible strategies whenever $\kappa\in\left[P,R\right]$ and $\chi\geqslant 1$ or $\chi\leqslant -\max\left(\frac{T-\kappa}{\kappa -S},\frac{\kappa -S}{T-\kappa}\right)$ \citep{hilbe:adaptive}.

Using the terminology of \citet{chenzinger}, we refer to strategies with $\chi>0$ as ``positively correlated'' (pcZD) strategies (although in generalized social dilemmas \citep{hauert:JTB:2006}, $\chi$ cannot take on values between $0$ and $1$). Important special cases of pcZD strategies in social dilemmas are so-called extortionate strategies ($\kappa =P$ and $\chi>1$) and generous strategies ($\kappa =R$ and $\chi >1$). The former ensure that $X$ receives a larger payoff than $Y$ beyond the mutual defection value, while the latter ensure that $Y$'s payoff is closer to the mutual cooperation value than $X$'s payoff is. Strategies with $\chi=1$, enforcing $\pi_{X}=\pi_{Y}$, are called fair, and include the well-known strategy tit-for-tat. In all generalized social dilemmas satisfying $S+T<2R$, $Y$'s optimal behavior against a pcZD strategy is effectively unconditional cooperation, which maximizes both players' payoffs. A pcZD strategy thus incentivizes $Y$ to cooperate even if doing so provides an unfair benefit to $X$. Depending on the particulars of $X$'s strategy, $Y$'s repeated cooperation can be implemented by various strategies, not just $\left(1,1,1,1\right)$. For instance, if $X$'s strategy satisfies $p_{CC}=1$, then $Y$ using a strategy with $q_{CC}=1$ would suffice to achieve repeated cooperation for most values of $\left(p_{CD},p_{DC},p_{DD}\right) ,\left(q_{CD},q_{DC},q_{DD}\right)\in\left[0,1\right]^{3}$ as long the Markov chain of \eq{markov_matrix} has a unique stationary distribution. Similarly, $\mathbf{q}=\left(1,1,q_{DC},q_{DD}\right)$ is equivalent to $\left(1,1,1,1\right)$ in most cases.

Another interesting kind of ZD strategy, particularly relevant for our results, is the so-called ``equalizer'' \citep{boerlijst:AMM:1997}, which allows $X$ to enforce $\pi_{Y}=K$ for certain $K\in\mathbb{R}$, regardless of $Y$'s strategy. \citet{pressdyson} show that if
\begin{subequations}
\begin{align}
p_{CD} &= \frac{p_{CC}\left(T-P\right) -\left(1+p_{DD}\right)\left(T-R\right)}{R-P} ; \\
p_{DC} &= \frac{\left(1-p_{CC}\right)\left(P-S\right) -p_{DD}\left(R-S\right)}{R-P} ,
\end{align}
\end{subequations}
then $X$ can ensure that $\pi_{Y}=\frac{\left(1-p_{CC}\right) P+p_{DD}R}{1-p_{CC}+p_{DD}}$. While equalizers (and ZD strategies more generally) represent a rather stringent form of payoff control, $X$ can also impose weaker conditions. For example, \citet{payoffcontrol} demonstrate how $X$ can choose its strategy to enforce linear inequalities involving $\pi_{X}$ and $\pi_{Y}$, such as simply setting an upper bound on $Y$'s payoff, $\pi_{Y}\leqslant K$.

Finally, perhaps the simplest way to visualize such attempts at payoff control is in terms of the feasible payoff region generated by a fixed strategy, $\mathbf{p}$, of $X$. For two players in an iterated game with parameters $\left(R,S,T,P\right)$, the space of all possible payoffs $\left(\pi_{Y},\pi_{X}\right)$ is the convex hull of the four points $\left(R,R\right)$, $\left(S,T\right)$, $\left(T,S\right)$, and $\left(P,P\right)$. However, once $X$'s strategy is fixed at $\mathbf{p}$ (while leaving $Y$ free to vary), the set of attainable payoffs is a subset of the full payoff region (\fig{feasible_regions}\textbf{a-e}). This feasible region, denoted $\mathcal{C}\left(\mathbf{p}\right)$, is the convex hull of at most 11 points \citep{mcavoy}. Controlling the payoffs amounts to constraining $\mathcal{C}\left(\mathbf{p}\right)$. If $\mathbf{p}$ is a ZD strategy, $\mathcal{C}\left(\mathbf{p}\right)$ is simply a line (\fig{feasible_regions}\textbf{d,e}).

\begin{figure}
\centering
\includegraphics[width=\linewidth]{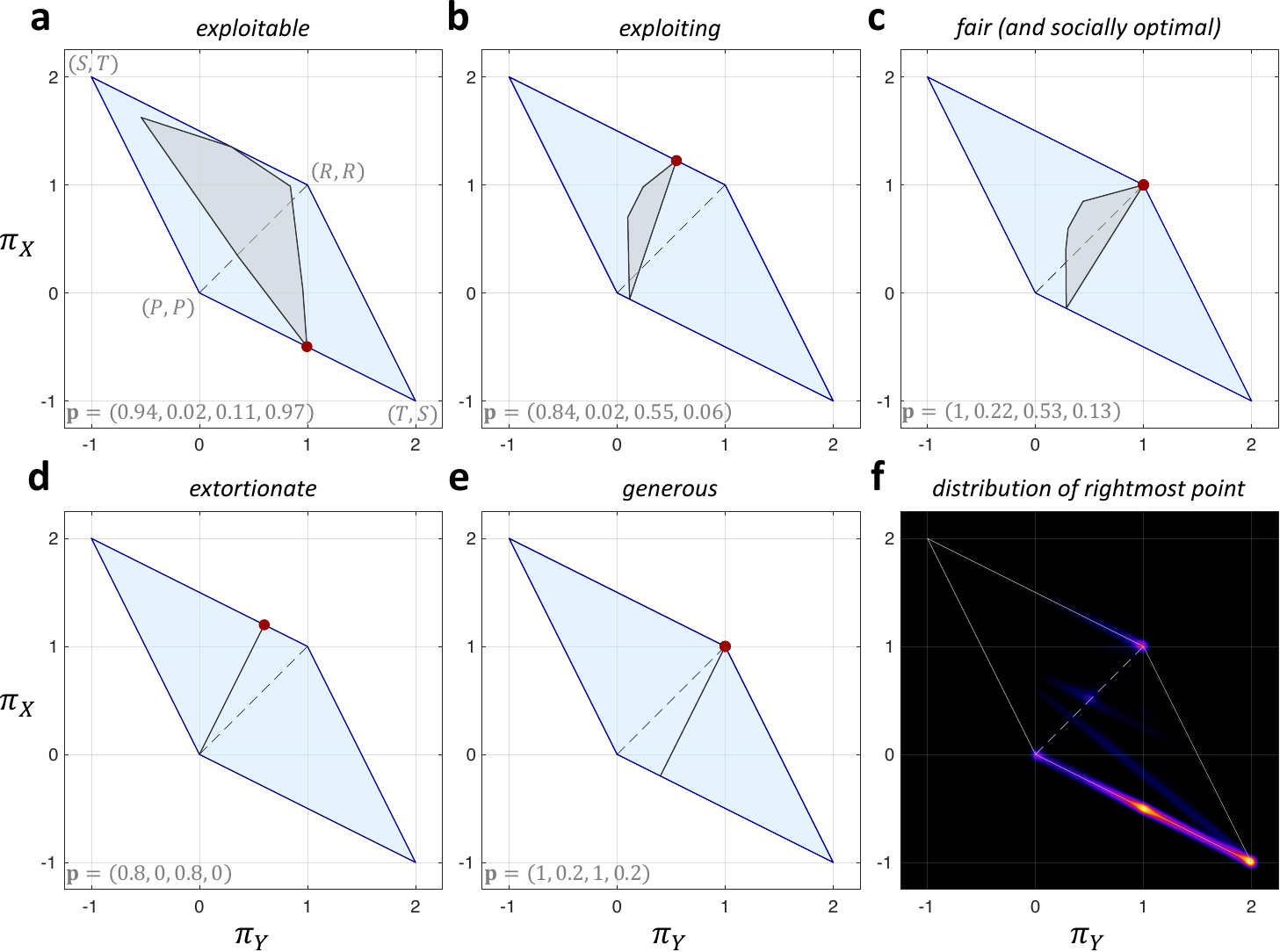}
\caption{\textbf{Feasible payoff regions in an iterated prisoner's dilemma.} The game associated with each panel is an iterated prisoner's dilemma with $\left(R,S,T,P\right) =\left(1,-1,2,0\right)$. In panels \textbf{a-e}, the blue region represents all possible payoffs for the game. The shaded region or solid line in the interior represents the feasible region $\mathcal{C}\left(\mathbf{p}\right)$ once $X$'s strategy $\mathbf{p}$ is fixed (while leaving $Y$ free to vary). The red point highlights the rightmost point of $\mathcal{C}\left(\mathbf{p}\right)$, i.e. the payoff pair corresponding to $Y$'s selfishly optimal strategy. In panels \textbf{a-c}, $X$'s strategies are general memory-one, and can be broadly characterized as ``exploitable,'' ``exploiting,'' or ``fair'' based on whether the rightmost point of $\mathcal{C}\left(\mathbf{p}\right)$ lies in the region $\pi_{X}<\pi_{Y}$, $\pi_{X}>\pi_{Y}$, or $\pi_{X}=\pi_{Y}$. Panel \textbf{d} shows the feasible region of an extortionate ZD strategy with $\kappa =P=0$, $\chi =2$, and $\phi =0.2$, which enforces the constraint $\pi_{X}=2\pi_{Y}$ (\eq{ZD_formula}). It is ``extortionate'' because the entire feasible line lies in the region $\pi_{X}\geqslant\pi_{Y}$, and any improvement in $Y$'s payoff provides an even greater increase in $X$'s payoff. Panel \textbf{e} depicts the feasible region of a generous ZD strategy with $\kappa =R=1$, $\chi =2$, and $\phi =0.2$, which enforces the constraint $1-\pi_{X}=2\left(1-\pi_{Y}\right)$. Providing a benign counterpart to extortion, generous strategies result in $\pi_{Y}\geqslant\pi_{X}$ for all strategies of $Y$. The rightmost point aligns $Y$'s selfish optimum with the fair and socially optimal payoffs of mutual cooperation. Finally, panel \textbf{f} depicts a heatmap representing the distribution of the rightmost point of $\mathcal{C}\left(\mathbf{p}\right)$ for $10^{6}$ memory-one strategies $\mathbf{p}$ of $X$, with each coordinate sampled independently from an arcsine distribution (which is used to efficiently explore the corners of the unit interval, $\left[0,1\right]$, and dates back to at least \citet{nowak:WSLS} in iterated games). Broadly speaking, our concern here is with whether a selfish $Y$ can actually make it to the red point.\label{fig:feasible_regions}}
\end{figure}

\subsection{Robustness of selfish optimization}\label{sec:robustness}
From $X$'s perspective, the use of such strategies is most effective when it influences $Y$'s actual behavior in the desired manner, which motivates one to study the robustness of these payoff control methods in practice. A common practical context is one in which the co-player $Y$ seeks to maximize its payoff by learning the best strategy, following an optimization procedure or learning rule. Indeed, interacting with such a selfish opponent concerned only with its own payoff, rather than shared benefit, is one of the situations where payoff control would be most desirable (for instance, a generous strategy in an IPD with $S+T<2R$ aligns $Y$'s own selfish optimum with the jointly beneficial payoff of mutual cooperation). Press and Dyson conjecture in \citep{pressdyson}, on the basis of numerical evidence, that when facing an extortionate ZD strategy in an IPD with $2P<S+T<2R$, the selfish learner $Y$ always has an evolutionary path leading to a globally payoff-maximizing strategy. The term ``evolutionary path'' in this context refers to a trajectory through $Y$'s strategy space, the unit four-dimensional hypercube, such that $Y$ is locally following the path of steepest ascent in payoff (with some learning rate).

\citet{chenzinger} formally prove a stronger version of this conjecture that applies to all pcZD strategies, not just extortionate ones, in any IPD with $2P<S+T<2R$ \citep{chenzinger}. They generalize Press and Dyson's notion of an ``evolutionary path'' to what they term an ``adapting path,'' formally defined in Section 3 of \citep{chenzinger}. Essentially, an adapting path is a smooth map from some time parameter to the hypercube, representing $Y$'s strategy as a function of time, such that \emph{(i)} $Y$'s payoff strictly increases with time and \emph{(ii)} in finite time, $Y$ reaches an endpoint that is at least locally optimal. Chen and Zinger's robustness result can be summarized as follows:
\begin{quote}
For a pcZD strategy of $X$ in an IPD with $2P\! <\! S+T\! <\! 2R$, all adapting paths of $Y$ lead to the maximum payoff attainable against $X$ (corresponding to cooperation by $Y$).
\end{quote}

\citet{chenzinger} note that the truth of the original conjecture of \citet{pressdyson}, regarding the existence of payoff-maximizing adapting paths for $Y$, does not provide $X$ enough assurance that $Y$ will behave in the desired manner. $Y$ could potentially follow some other path leading to a local (but not global) optimum. Their proof allays these worries in any IPD with $2P<S+T<2R$ by establishing that \textit{all} adapting paths for $Y$ will reach the desired behavior.

However, once the restriction $2P<S+T<2R$ is removed, their proof no longer holds, as we discuss in Appendix A. In the next section, we demonstrate through counterexamples that no other proof would suffice; the robustness result itself simply does not extend to more general cases. We demonstrate how robustness can fail to hold for pcZD strategies in IPD games with $S+T>2R$ or $S+T>2P$. We also consider the general memory-one case where $X$ is no longer using a ZD strategy.

\section{Methods and results}\label{sec:methods_and_results}
The learning rule we use for the selfish player $Y$ is projected gradient ascent (PGA), a local hill-climbing algorithm \citep{haskell:QAM:1944}. This is a natural choice, given its simplicity and widespread use in optimization. Press and Dyson's conjecture is based on numerical simulations using versions of gradient ascent \citep{pressdyson}. Furthermore, this type of localized algorithm is particularly useful in more general contexts. If $Y$ is aware that $X$ is using a pcZD strategy in an IPD with $S+T<2R$, then $Y$ simply has to repeatedly cooperate to achieve the maximum payoff. However, in situations where $S+T>2R$, $Y$'s optimal strategy is not generally repeated cooperation. Rather, it depends on the specific implementation of $X$'s pcZD strategy, which includes not only the slope, $\chi$, but also $\phi$, which $Y$ might not have access to. In such settings, local information, including an estimated payoff gradient at the current strategy, becomes all the more useful for optimization. This motivation is relevant even in IPDs with $S+T<2R$ when $X$ is using a general memory-one rather than a pcZD strategy, in which case $Y$'s optimal strategy is not necessarily known from the outset.

As a selfish learner, $Y$'s only goal is to optimize $\pi_{Y}$. When the game parameters, $\left(R,S,T,P\right)$, and $X$'s strategy, $\mathbf{p}$, are fixed, the only free parameters on which $\pi_{Y}$ depends are those of $Y$'s strategy, $\mathbf{q}$. Therefore, against a fixed opponent in an undiscounted game, the parameter space for optimization is the four-dimensional unit hypercube, $\left[0,1\right]^{4}$, when $Y$ is a memory-one player. Qualitatively, $Y$ begins at a random initial strategy, and at each time step takes a small ``step'' in the direction of the gradient at its current location. Since $\pi_{Y}\left(\mathbf{p},\mathbf{q}\right)$ is defined over a compact parameter space, we ``project'' $Y$'s strategy to the closest point in the domain after each gradient step, ensuring that the new strategy components are valid probabilities \citep{calamai:MP:1987}. Given an initial strategy, $\mathbf{q}_{0}$, and a scalar learning rate, $\eta$, projected gradient ascent (PGA) is defined by the update rule
\begin{align}
\mathbf{q}_{n+1} = \textrm{proj}\left( \mathbf{q}_{n} + \eta \nabla_{\mathbf{y}} \pi_{Y}\left(\mathbf{p},\mathbf{y}\right)\vert_{\mathbf{y}=\mathbf{q}_{n}} \right) ,
\end{align}
where $\nabla$ denotes the gradient operator, and the function $\textrm{proj}\left(\mathbf{x}\right)$ projects a point $\mathbf{x}\in\mathbb{R}^{m}$ onto the unit $m$-dimensional hypercube $\left[0,1\right]^{m}$. Componentwise, for each $i=1,\dots ,m$, $\textrm{proj}\left(x_{i}\right) =1$ if $x_{i}>1$, $\textrm{proj}\left(x_{i}\right) =x_{i}$ if $0\leqslant x_{i}\leqslant 1$, and $\textrm{proj}\left(x_{i}\right) =0$ if $x_{i}<0$.

The key question relevant to robustness is the following: after sufficiently many learning steps, can we be assured that $Y$'s final ``optimized'' strategy globally maximizes $\pi_{Y}$? Before demonstrating that we cannot have this assurance, we first establish some general properties of stable points of PGA. Due to the projection operator, there are two types of stable points: those with vanishing gradient and those without. For the latter, the gradient at a boundary point can push $Y$'s strategy out of the hypercube, only to be projected back to the same boundary point.

We first explore the conditions under which the gradient vanishes. For fixed strategies $\mathbf{p}$ and $\mathbf{q}$, consider the vector space of game parameters for which the gradient of $\pi_{Y}$ vanishes at $\mathbf{q}$,
\begin{align}
V\left(\mathbf{p},\mathbf{q}\right) &\coloneqq \left\{\begin{pmatrix}R \\ S \\ T \\ P\end{pmatrix}\in\mathbb{R}^{4}\ :\ \nabla_{\mathbf{y}}\pi_{Y}\left(\mathbf{p},\mathbf{y}\right)\vert_{\mathbf{y}=\mathbf{q}} =0\right\} . \label{eq:V_def}
\end{align}
The vector space of game parameters for which $\mathbf{p}$ acts as an equalizer strategy is
\begin{align}
E\left(\mathbf{p}\right) &\coloneqq \left\{\begin{pmatrix}R \\ S \\ T \\ P\end{pmatrix}\in\mathbb{R}^{4}\ :\ \pi_{Y}\left(\mathbf{p},\mathbf{q'}\right)\textrm{ is independent of }\mathbf{q'}\in\left[0,1\right]^{4}\right\} ,
\end{align}
which is naturally a subspace of $V\left(\mathbf{p},\mathbf{q}\right)$. Our first result concerns when these two spaces coincide:
\begin{theorem}\label{thm:equalizer}
If $\mathbf{p}\in\left[0,1\right]^{4}\setminus\left\{\left(1,1,0,0\right)\right\}$, then $V\left(\mathbf{p},\mathbf{q}\right) =E\left(\mathbf{p}\right)$ whenever $\mathbf{q}\in\left(0,1\right)^{4}$.
\end{theorem}
This result means that for any strategy $\mathbf{p}$ different from ``repeat,'' if $\nabla_{\mathbf{y}}\pi_{Y}\left(\mathbf{p},\mathbf{y}\right)\vert_{\mathbf{y}=\mathbf{q}} =0$ for \emph{some} $\mathbf{q}\in\left(0,1\right)^{4}$, then $\mathbf{p}$ is an equalizer strategy and $\nabla_{\mathbf{y}}\pi_{Y}\left(\mathbf{p},\mathbf{y}\right)\vert_{\mathbf{y}=\mathbf{q}} =0$ for \emph{every} $\mathbf{q}\in\left(0,1\right)^{4}$. Therefore, a purely stochastic strategy (in the interior of the hypercube, $\mathbf{q}\in\left(0,1\right)^{4}$) does not yield a locally or globally optimal payoff, and thus cannot be an endpoint of gradient ascent, unless the opponent is playing an equalizer strategy. In fact, we can say more:
\begin{proposition}\label{prop:endpoints}
Suppose $X$ and $Y$ are general memory-one players in any iterated symmetric two-player two-action game, $\mathbf{p}$ is a fixed strategy of $X$, and $\mathbf{q}_{\text{final}}$ is $Y$'s optimized strategy following a projected gradient trajectory starting from any initial strategy, $\mathbf{q}_{0}$. Unless $\mathbf{p}$ is an equalizer strategy (the degenerate case corresponding to a flat payoff landscape), $\mathbf{q}_{\text{final}}$ is constrained as follows:
\begin{enumerate}

\item[(a)] For all $\mathbf{p}\in\left[0,1\right]^{4}\setminus\left\{\left(1,1,0,0\right)\right\}$, $\mathbf{q}_{\text{final}}$ has at least one deterministic component;

\item[(b)] If $\mathbf{p}$ is randomly sampled from $\left[0,1\right]^{4}$ (with each of the four components independent), then with probability one, $\mathbf{q}_{\text{final}}$ is one of the following: $\left(1,1,q_{DC},q_{DD}\right)$, $\left(q_{CC},q_{CD},0,0\right)$, or a strategy with all four components deterministic.

\end{enumerate}
\end{proposition}
Note that when the game is a social dilemma, $\left(1,1,q_{DC},q_{DD}\right)$ and $\left(q_{CC},q_{CD},0,0\right)$ may be interpreted as repeated cooperation and repeated defection, respectively. However, we emphasize that the results above apply to any symmetric $2\times 2$ game.

The proofs of \thm{equalizer} and \prop{endpoints} are provided in Appendix B. It follows that under PGA, $Y$ is guaranteed to make it to the boundary of the hypercube regardless of its initial strategy (unless $X$ is using an equalizer strategy, in which case optimization is irrelevant). This leaves open the possibility that $Y$ reaches a locally but not globally optimal endpoint on the boundary.

Below, we demonstrate using examples from IPD games that this possibility often does occur in practice. We use a learning rate of $\eta = 10^{-2}$ and apply the strategy update rule until the magnitude of successive changes in $Y$'s payoff fall below $10^{-15}$, indicating that the algorithm has effectively converged to an endpoint. Since $Y$'s payoff is generically of the same order as $1$, and we perform computations with $16$ digits of precision, smaller changes in payoff cannot be accurately detected in practice.

\subsection{When $X$ uses a pcZD strategy}
Here, we consider the case $\kappa =P$. As a result, $X$'s strategy, $\mathbf{p}$, is parametrized in terms of $\chi$ and $\phi$ (\eq{ZD_formula}). \fig{GA_trajectories} illustrates $Y$'s PGA optimization for three different initial strategies, against a fixed strategy of $X$, $\mathbf{p} = \left(1,0.12,0.88,0\right)$, in an IPD with game parameters $\left(2,-1,7,0\right)$. In particular, these payoffs satisfy $S+T>2R$, so this game is outside of the purview of the robustness result of \citet{chenzinger}. Each of $Y$'s initial strategies leads to a different final strategy under PGA, corresponding to three distinct ``optimized'' values of $\pi_{Y}$. The strategy in \fig{GA_trajectories}\textbf{a} results in repeated mutual cooperation, which provides the lowest payoffs to both players relative to those of \fig{GA_trajectories}\textbf{b,c}. The latter two strategies allow for both $CD$ and $DC$ to occur in the long run, and alternation between these two states results in higher mutual benefit than repeated cooperation in IPDs with $S+T>2R$. While neither of these two cases yields pure alternation (i.e. $CD$ on odd rounds, $DC$ on even rounds, with $CC$ and $DD$ being eliminated in the long run), they are still superior to mutual cooperation. The strategy of \fig{GA_trajectories}\textbf{c} globally maximizes $\pi_{Y}$ (as well as $\pi_{X}$, since $\mathbf{p}$ is a pcZD strategy and thus the goals are aligned).

\begin{figure}
\centering
\includegraphics[width=\linewidth]{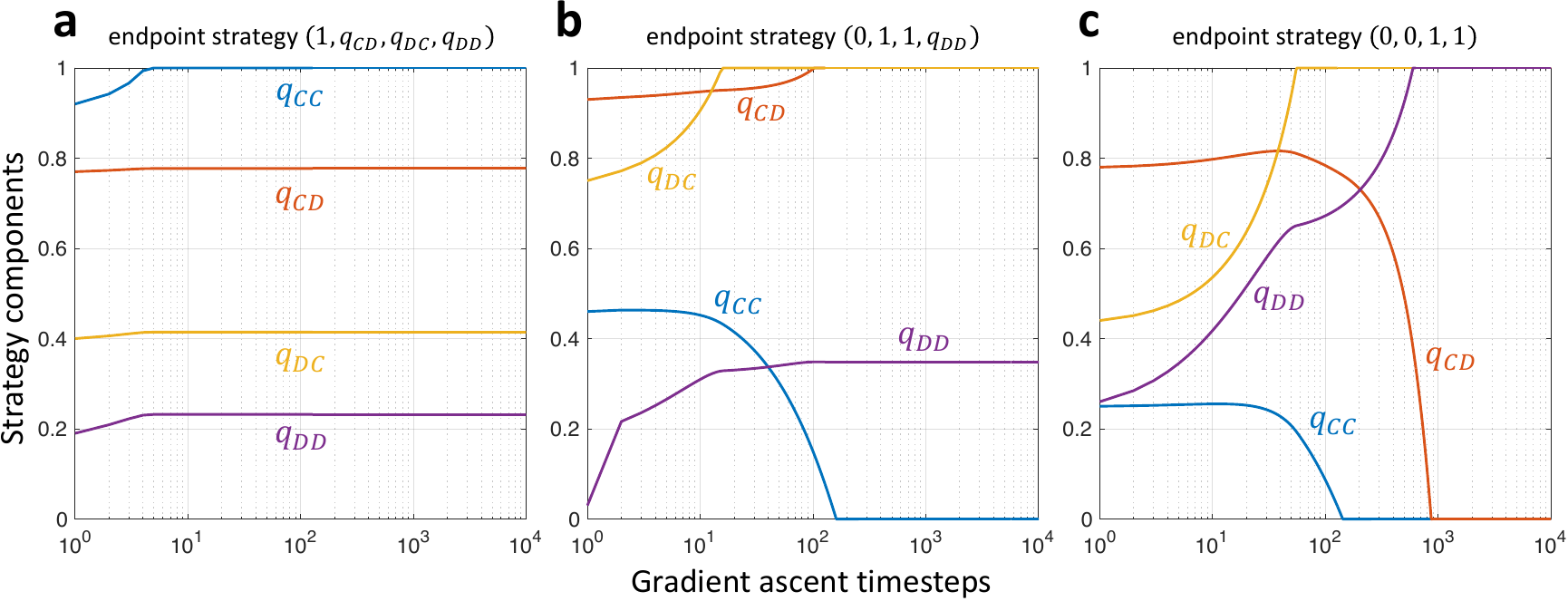}
\caption{\textbf{Gradient ascent trajectories leading to three distinct final strategies against a fair ZD opponent in an iterated prisoner's dilemma.} The game parameters are $\left(R,S,T,P\right) =\left(2,-1,7,0\right)$, satisfying $S+T>2R$. The ZD player $X$ uses strategy $\left(1,0.12,0.88,0\right)$ with $\chi =1$ and $\phi =0.11$. Each panel represents the updating of $Y$'s strategy at each time step of gradient ascent. The three different initial conditions shown here lead to distinct endpoint strategies, listed in ascending order of payoffs. Strategy \textbf{a} is effectively equivalent to unconditional cooperation: once $q_{CC}$ reaches $1$, repeated mutual cooperation becomes the Markov stationary state because $p_{CC}$ is also 1. In this situation, the values of $q_{CD}, q_{DC}, q_{DD}$ are irrelevant, so their corresponding gradient components vanish, and they stop updating. This strategy is not globally optimal in the IPD when $S+T>2R$. Strategy \textbf{c} is a pure alternator, with $D$ following $C$ and vice versa regardless of the co-player's action. When $Y$ uses this strategy, all four states of the Markov chain are in the support of the stationary distribution, with $CD$ and $DC$ most frequently represented (when the game is at either of these states, the probability of transitioning to the complementary state is 0.88). Strategy \textbf{b} also facilitates alternating behavior; however, the mechanics of how it cycles through the Markov states differs from that of case \textbf{c}, with mutual cooperation occurring more frequently than it does with strategy \textbf{c}, thereby resulting in a lower payoff. With strategy \textbf{b}, mutual defection cannot occur in the stationary state, so the value of $q_{DD}$ is irrelevant.\label{fig:GA_trajectories}}
\end{figure}

Given the existence of multiple endpoints (in both strategy and payoff space) under PGA, it would be useful to know how often the suboptimal ones occur when $Y$'s initial strategy is randomly sampled. The more frequent they are, the less likely that $X$'s attempt at payoff control by using a ZD strategy will have the desired effect. We emphasize that various strategies of $Y$ can correspond to the same payoff, so we are interested in the relative frequencies of distinct payoff-space endpoints $\left(\pi_{Y},\pi_{X}\right) \in \mathcal{C}\left(\mathbf{p}\right)$ rather than strategy-space endpoints. For example, \fig{GA_trajectories}\textbf{a} yields mutual cooperation because $q_{CC}=p_{CC}=1$, so all strategies $\mathbf{q}=\left(1,q_{CD},q_{DC},q_{DD}\right)$ with a well-defined limiting distribution against $\mathbf{p}$ correspond to just a single payoff endpoint.

To estimate the relative frequencies of distinct payoff endpoints against $\mathbf{p}$, we sample $10^{5}$ initial strategies for $Y$, with each coordinate chosen independently from an arcsine distribution. We run PGA on each of these while keeping $\mathbf{p}$ fixed at $\left(1,0.12,0.88,0\right)$. The distribution of payoff endpoints is plotted as a heatmap in \fig{distribution_ZD}\textbf{a}. While the majority ($84{,}247$) of initial strategies $\mathbf{q}_{0}$ lead to the globally optimal payoff, a non-negligible number lead to the suboptimal endpoints ($4{,}728$ and $11{,}025$ in descending order of payoffs). The same phenomenon, of distinct payoff endpoints against a fixed ZD strategy of $X$, arises in the IPD when $S+T<2P$ as well, which is illustrated in \fig{distribution_ZD}\textbf{b} with $10^{5}$ initial strategies for $Y$ sampled from an arcsine distribution.

\begin{figure}
\centering
\includegraphics[width=0.9\linewidth]{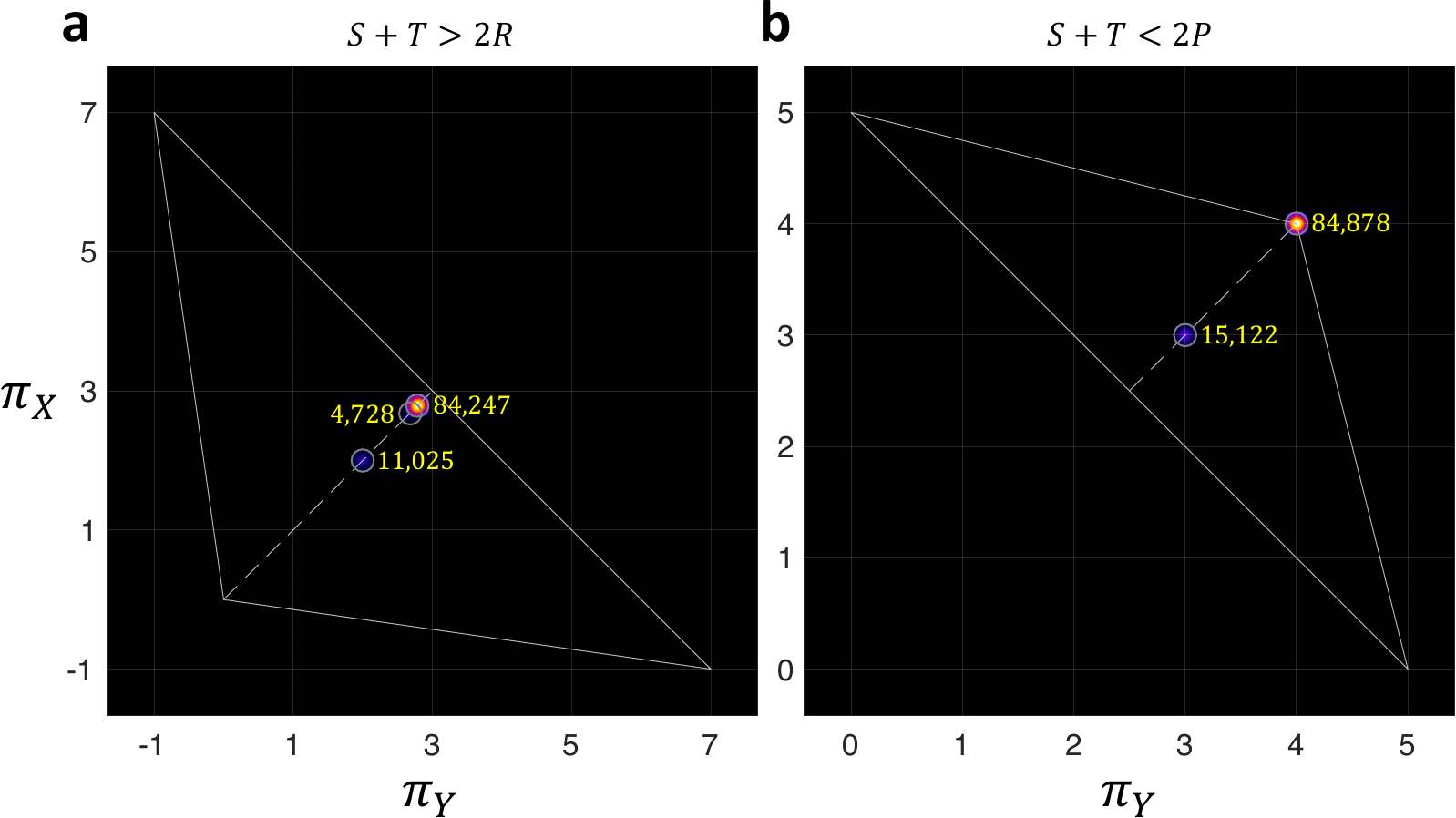}
\caption{\textbf{Distribution of final payoffs against ZD opponents in IPDs not satisfying $2P<S+T<2R$.} In each case, $X$'s strategy, $\mathbf{p}$, is a fixed, fair ZD strategy while $Y$ optimizes using projected gradient ascent. The highlighted points correspond to the final payoffs $\left(\pi_{Y}\left(\mathbf{p},\mathbf{q}_{\text{final}}\right),\pi_{X}\left(\mathbf{p},\mathbf{q}_{\text{final}}\right)\right)$. The heatmaps are plotted for the final payoffs corresponding to $10^{5}$ initial strategies of $Y$, where each coordinate is chosen independently from an arcsine distribution. Panel \textbf{a}, IPD with $\left(R,S,T,P\right) =\left(2,-1,7,0\right)$ (the same parameters and endpoint strategies depicted in \fig{GA_trajectories}). $X$'s strategy is $\mathbf{p}=\left(1,0.12,0.88,0\right)$, with $\chi =1$ and $\phi =0.11$. The unique payoff endpoints are $\left(2,2\right)$, $\left(2.67,2.67\right)$, and $\left(2.79,2.79\right)$, described qualitatively in \fig{GA_trajectories}. Panel \textbf{b}, IPD with $\left(R,S,T,P\right) =\left(4,0,5,3\right)$. $X$'s strategy is $\mathbf{p}=\left(1,0.85,0.15,0\right)$, with $\chi =1$ and $\phi =0.03$. The unique payoff endpoints correspond to defection by $Y$ (point $\left(3,3\right)$) and cooperation by $Y$ (point $\left(4,4\right)$). In both figures, the highest payoff values attained are the maximum attainable. \textit{Game parameters and strategies above are exact; payoff values have been rounded to two decimal places.}\label{fig:distribution_ZD}}
\end{figure}

The examples in \fig{distribution_ZD} each focus on a single strategy, $\mathbf{p}$, of $X$. Next, we consider how many pcZD strategies show similar behavior, i.e. have a payoff landscape that allows for distinct endpoints, each with nontrivial basins of attraction. Our simulations indicate that such landscapes are the norm rather than the exception. We consider $100$ pcZD strategies $\mathbf{p}$ parametrized in terms of $\chi$ and $\phi$, selected as follows: for each $\chi = 1,2,\dots ,20$, we take five linearly spaced values of $\phi$ between $10^{-4}$ and $0.99\phi_{\max}$, where $\phi_{\max}$ is the largest allowable value of $\phi$ for that $\chi$. Against each $\mathbf{p}$, we run PGA on $100$ random initial strategies $\mathbf{q}_{0}$ of $Y$. For the payoffs $\left(R,S,T,P\right) =\left(2,-1,7,0\right)$ from before, satisfying $S+T>2R$, we find that multiple endpoints are discovered in $71$ out of $100$ of these strategies $\mathbf{p}$. In each such case, there are exactly two distinct endpoints (rounded to two decimal points). The appearance of such endpoints is particularly noteworthy in these simulations: the number of samples of $\mathbf{q}_{0}$ is small, just $100$, yet they lead to suboptimal endpoints with non-negligible frequency, averaging $\approx\! 8\%$ among those $71$ strategies. With $\left(R,S,T,P\right) =\left(3,-1,9,0\right)$ (which still satisfies $S+T>2R$), we find that multiple endpoints are discovered in $78$ runs, and among these runs the suboptimal outcome is found $\approx\! 15\%$ of the time, on average. The results are even more striking when $S+T<2P$. With game parameters $\left(R,S,T,P\right) =\left(4,0,5,3\right)$, multiple endpoints are discovered in all $100$ runs. Notably, for each $\mathbf{p}$, there are always trajectories that lead to the payoff for mutual defection, $P$. This suboptimal outcome has frequency averaging $\approx\! 15\%$.

In all runs of our iterated prisoner's dilemma with $S+T>2R$ or $S+T<2P$, the globally maximizing payoff point $\left(\pi_{Y}^{\max},\pi_{X}^{\max}\right)$ always appears and has a higher frequency than the suboptimal endpoint. When $S+T>2R$, this globally optimal point depends on $\phi$ in addition to $\chi$. (This is due to the manner in which violation of $S+T<2R$ affects the boundaries of the total payoff region; in particular, repeated cooperation is no longer automatically a vertex of $\mathcal{C}\left(\mathbf{p}\right)$.) For a given slope, $\chi$, different values of $\phi$ result in numerically different optimal payoffs, and therefore a different length of the (linear) feasible region, $\mathcal{C}\left(\mathbf{p}\right)$. Unsurprisingly, this dependence on $\phi$ is not present when $S+T<2P$ since the optimal payoff can still be achieved through cooperation.

\subsection{When $X$ uses a general memory-one strategy}
As noted in Appendix A of \citet{chenzinger} and shown visually in our \fig{feasible_regions}, when $X$'s strategy is no longer ZD, $Y$'s payoff-maximizing strategy need not be repeated cooperation even if $2P<S+T<2R$. Once the feasible region is no longer linear, it is perhaps less tempting to assume that $Y$ always reaches the global maximum (even in the usual IPD with $S+T<2R$), although this assumption does crop up in the literature \citep{payoffcontrol}. Investigating the general case is important because ZD strategies impose a particularly strict requirement on the shape of $\mathcal{C}\left(\mathbf{p}\right)$. $X$ may desire more flexibility in how it constrains $\mathcal{C}\left(\mathbf{p}\right)$ and choose a general memory-one strategy with the desired properties. Even in contexts besides payoff control, it is beneficial for player $Y$ to know about the reliability of selfish optimization when faced with a general memory-one (or even memory-$n$) player, $X$.

When $X$'s strategy $\mathbf{p}$ is memory-one but not ZD, we are no longer within the purview of the robustness result of \citet{chenzinger}. In this case, for simplicity, we revert to the common assumption that $2P<S+T<2R$ and use the game parameters $\left(3,0,5,1\right)$. Here, too, we find that payoff landscapes allow for multiple endpoints under PGA, with no more than three distinct endpoints (when rounded to two decimal points) in any of the payoff landscapes corresponding to $1{,}000$ arcsine-sampled strategies $\mathbf{p}$. Unlike in the case of ZD strategies when either $S+T>2R$ or $S+T<2P$, such landscapes appear to be in the minority, yet they occur often enough to be notable. Moreover, when such landscapes do occur, they are typically more pathological than those corresponding to ZD strategies. Suboptimal endpoints occur with higher relative frequency and can be further away from true optimality. \fig{distribution_general} shows two strategies $\mathbf{p}$ representing particularly pathological cases. In both examples of \fig{distribution_general}, the \textit{least} rewarding endpoint attained by $Y$ has the \textit{highest} relative frequency ($\approx\! 61\%$ in \fig{distribution_general}\textbf{a} and $\approx\! 44\%$ in \fig{distribution_general}\textbf{b}).

\begin{figure}
\centering
\includegraphics[width=0.9\linewidth]{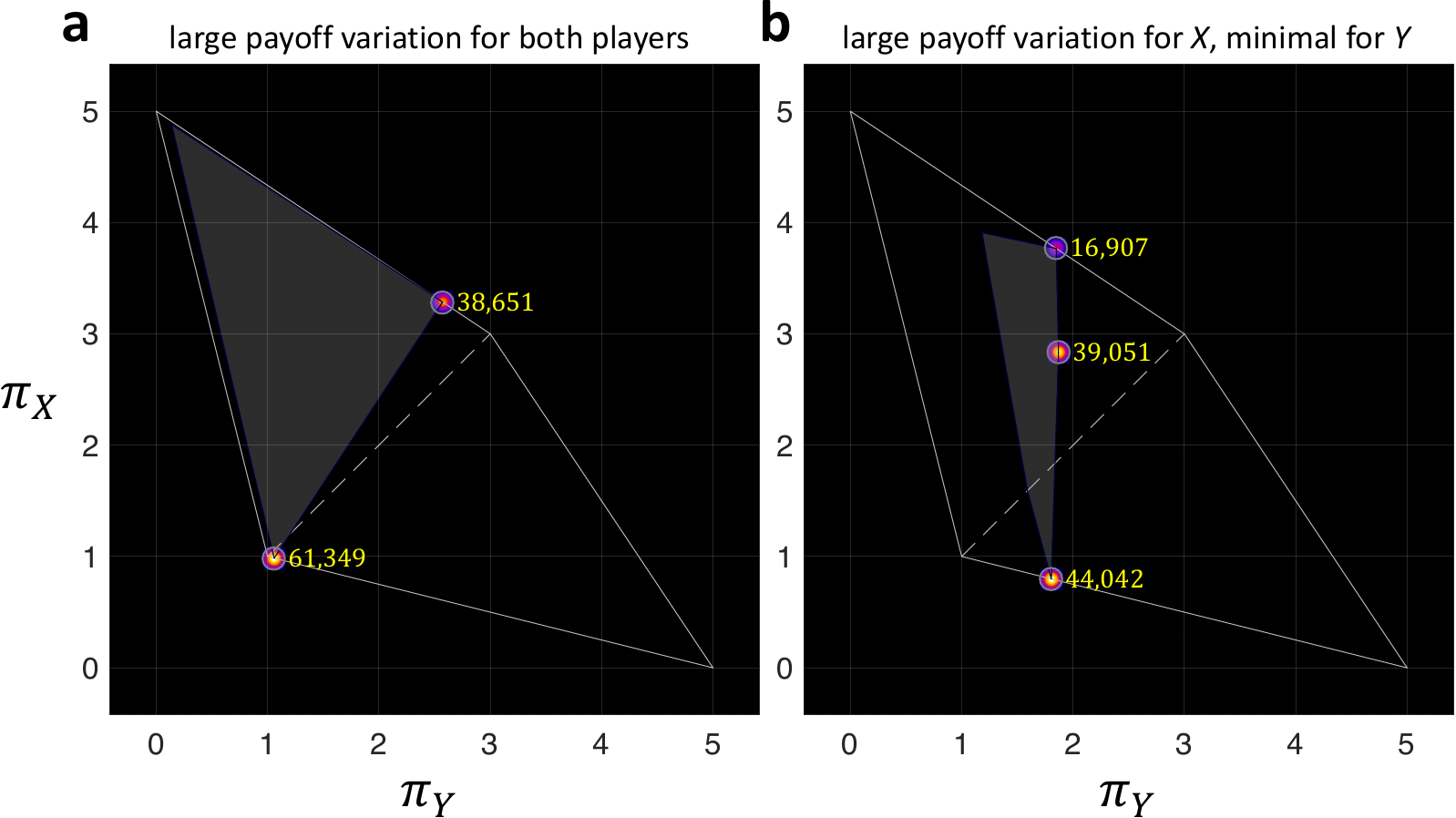}
\caption{\textbf{Distribution of final payoffs against general memory-one opponents in an iterated prisoner's dilemma.} In both panels, the game parameters are $\left(R,S,T,P\right) =\left(3,0,5,1\right)$, satisfying $2P<S+T<2R$. $X$ uses a non-ZD memory-one strategy, which, in each example, is fixed throughout. The highlighted points correspond to the final payoffs $\left(\pi_{Y}\left(\mathbf{p},\mathbf{q}_{\text{final}}\right),\pi_{X}\left(\mathbf{p},\mathbf{q}_{\text{final}}\right)\right)$ obtained after $Y$ optimizes using gradient ascent. The heatmaps are plotted for the final payoffs corresponding to $10^{5}$ initial strategies of $Y$, with each coordinate sampled independently from an arcsine distribution. In \textbf{a}, $X$'s strategy is $\mathbf{p}=\left(0.997,0.005,0.018,0.015\right)$. The payoff endpoints correspond to $Y$'s repeated defection (point $\left(1.06,0.99\right)$) and repeated cooperation (point $\left(2.57,3.29\right)$). The differences in the payoffs between these two points is quite large for both players, $1.51$ for $Y$ and $2.30$ for $X$. In \textbf{b}, $X$'s strategy is $\mathbf{p}=\left(0.860,0,0.225,0.252\right)$. The payoff endpoints correspond to $Y$'s repeated defection (point $\left(1.81,0.80\right)$), repeated cooperation (point $\left(1.85,3.77\right)$, and win-stay, lose-shift (point $\left(1.87,2.83\right)$). While the difference in $Y$'s payoffs across these points is relatively small, the range of $X$'s payoffs is $2.97$, which is nearly $R-S=3$, the difference to $X$ between the mutually cooperative state $CC$ and the state $CD$ in which $X$ is fully exploited by $Y$. In both panels, the maximum attainable payoff for $Y$ is attained, but the least optimal endpoint attained by $Y$ has the highest relative frequency. \textit{Game parameters and strategies above are exact; payoff values have been rounded to two decimal places.}\label{fig:distribution_general}}
\end{figure}

Furthermore, the endpoint separations in \fig{distribution_general} are quite large. The distance between endpoints has important implications for both players. From $Y$'s perspective, a larger distance makes the outcome of selfish optimization more uncertain. In \fig{distribution_general}\textbf{a}, both players are more likely to receive a significantly suboptimal payoff when $Y$ optimizes selfishly. The implications for $X$ are most vivid in \fig{distribution_general}\textbf{b}, where $\pi_{Y}$ is nearly the same at all three endpoints, but $\pi_{X}$ varies significantly. The range of $\pi_{X}$ is nearly $R-S$, the difference between the mutually cooperative state $CC$ and the state $CD$ where $X$ is exploited by $Y$. Selfish optimization would thus safely bring $Y$ close to optimality and leave $X$'s outcome vulnerable to $Y$'s initial strategy. From the perspective of payoff control, strategies like the ones depicted in \fig{distribution_general} would be highly undesirable for $X$.

Stated slightly differently, the same strategy of $X$ can appear qualitatively different to $Y$ depending on the latter's initial strategy. We can broadly characterize $X$ as ``exploitable,'' ``exploiting,'' or ``fair'' based on whether the rightmost point of $\mathcal{C}\left(\mathbf{p}\right)$ (i.e. $Y$'s optimal point) lies in the region $\pi_{X}<\pi_{Y}$, $\pi_{X}>\pi_{Y}$, or $\pi_{X}=\pi_{Y}$ (\fig{feasible_regions}). When $Y$'s selfish optimization results in multiple endpoints, they need not all lie in the same region, as shown in \fig{distribution_general}. Therefore, while we believe that the geometric characterization depicted in \fig{feasible_regions} is a meaningful way to think about $X$'s strategy, the underlying assumption is that a selfish opponent $Y$ uses a sufficiently sophisticated method to attain the optimal (red) payoff point. Gradient ascent is evidently not such a method.

\subsection{Variations on the model and future directions}

\subsubsection{Noisy optimization and random search}
Moving forward, one might hope to improve the learning rule by adding noise, rather than deterministically following the payoff gradient. Preliminary simulations suggest that noisier procedures are, on the whole, more successful at finding global optima. However, significant room for improvement remains and the same qualitative issues arise as in the deterministic case. For these simulations, the learning rule we use is local random search (LRS), which is defined as follows. The selfish player $Y$ first samples a new strategy by choosing each coordinate $i$ uniformly in the interval $\left(q_{i}-\varepsilon ,q_{i}+\varepsilon\right)$ for small $\varepsilon >0$, truncating at $0$ and $1$ so that the sampled strategy remains viable. If this strategy yields a higher payoff, $Y$ switches to the new strategy. (In practice, we need the difference in payoffs from the old strategy to the new strategy to be greater than $10^{-15}$ in order to detect this payoff increase.) A larger sampling neighborhood corresponds to more noise. Given the stochasticity of sampling with LRS, this projected strategy is sometimes $\left(1,1,0,0\right)$, at which the Markov stationary distribution is not unique (a problem we do not encounter with PGA updating). We avoid this issue by using a discounting factor of $\lambda = 0.9999 < 1$ to approximate an iterated game with infinite time horizon. We consider the optimization process to have terminated when there has been no update in $10^{4}$ steps.

For $100$ general memory-one strategies $\mathbf{p}$ of $X$, each of which was previously found to result in multiple endpoints when $Y$ uses PGA, we consider the performance of LRS with increasing amounts of noise. We use ten equally-spaced values of $\varepsilon$ from $10^{-2}$ to $0.5$ (inclusive). For each $\mathbf{p}$, we run LRS for all ten values of $\varepsilon$ on a set of $1{,}000$ initial $Y$ strategies. Our performance metric for each $\varepsilon$ is the expected payoff $\left<\pi_{Y}\right>_{Y}$, averaged over all $1{,}000$ runs of $Y$.

For each strategy, we thus increase the amount of noise, represented by $\varepsilon$, by $\left(0.5-10^{-2}\right) /9$ a total of nine times. Across all $100$ strategies for $X$, this results in $900$ total noise increments. We find that $\approx 49\%$ of these noise increments result in higher values of $\left<\pi_{Y}\right>_{Y}$, while $\approx 51\%$ result in a non-positive change. Approximately $16\%$ of noise increments actually decrease $\left<\pi_{Y}\right>_{Y}$. This negative behavior is not restricted to a few strategies of $X$ with ``bad'' payoff landscapes; $55\%$ of $X$'s strategies show one or two instances of $\left<\pi_{Y}\right>_{Y}$ decreasing with a noise increment, while $22\%$ of $X$'s strategies show three or more instances of this pathology. Therefore, while noise statistically improves performance overall, this cannot be guaranteed in individual cases and the reliability of optimization is not significantly improved by LRS.

\subsubsection{Trembling hands and implementation errors}
A more promising variation on the model is to include noise in the implementation of $Y$'s strategy, rather than in the learning rule. We model implementation errors as follows. Suppose $Y$ intends to use a memory-one strategy $\mathbf{q}=\left(q_{CC},q_{CD},q_{DC},q_{DD}\right)$. However, $Y$ has a ``trembling hand,'' so with probability $\varepsilon >0$, $Y$ cooperates with probability $1-q_{xy}$ instead of $q_{xy}$. In this scenario, $Y$'s intended strategy remains $\textbf{q}$, but its effective strategy is $\left(1-\varepsilon\right)\mathbf{q}+\varepsilon\left(\mathbf{1}-\mathbf{q}\right)$ \citep{hilbe:JTB:2015}.

In order to evaluate the effects of implementation errors on the optimization process, we assume that only the optimizing player ($Y$) is susceptible to a trembling hand. Both players (fixed, $X$, and optimizing, $Y$) could make errors, in principle, but this would then change the fitness landscape that $Y$ explores. During optimization, $Y$ varies the parameters of its intended strategy $\mathbf{q}$ using gradient ascent, but the payoff $\pi_{Y}\left(\mathbf{p},\left(1-\varepsilon\right)\mathbf{q}+\varepsilon\left(\mathbf{1}-\mathbf{q}\right)\right)$ at each step is calculated using $X$'s intended strategy, $\mathbf{p}$, and $Y$'s effective strategy, $\left(1-\varepsilon\right)\mathbf{q}+\varepsilon\left(\mathbf{1}-\mathbf{q}\right)$.

Note that for any $\varepsilon \in \left(0,1\right)$, $Y$'s effective strategy is fully mixed, even if its intended strategy $\mathbf{q}$ has deterministic components. Prima facie this provides a substantial improvement in robustness: local optima of the form ``always cooperate'' or ``always defect,'' where the deterministic components of $\mathbf{q}$ make the payoff $\pi_{Y}\left(\mathbf{p},\mathbf{q}\right)$ independent of the remaining components, are no longer stable points of PGA with the modified payoff function $\pi_{Y}\left(\mathbf{p},\left(1-\varepsilon\right)\mathbf{q}+\varepsilon\left(\mathbf{1}-\mathbf{q}\right)\right)$. Our numerical simulations confirm this improvement but also reveal that it might only be achievable over sufficiently long time horizons, with PGA trajectories spending considerable time in the vicinity of what would be a suboptimal endpoint.

We run simulations in each of the three types of IPD games, using game parameters $\left(2,-1,7,0\right)$ for $S+T>2R$, $\left(3,0,5,1\right)$ for $2P < S+T < 2R$, and $\left(4,0,5,3\right)$ for $S+T<2P$. In each game, we consider $1{,}000$ arcsine-sampled memory-one strategies $\mathbf{p}$ of $X$. For each of these, we run PGA starting from $500$ arcsine-sampled intended strategies $\mathbf{q}$ of $Y$. We find that in all cases, the endpoint of $Y$'s trajectory globally maximizes its payoff (rounded to two decimal points). However, for trajectories tending to suboptimal points when there are no errors, the addition of errors does not always change the learning process over short time horizons. In \fig{trajectories_with_errors}, we have an example of a trajectory that leads to a small neighborhood of a suboptimal point, remains there for an extended period of time (making small gradient steps), and eventually escapes toward a much better outcome for both players.

\begin{figure}
	\centering
	\includegraphics[width=1.0\linewidth]{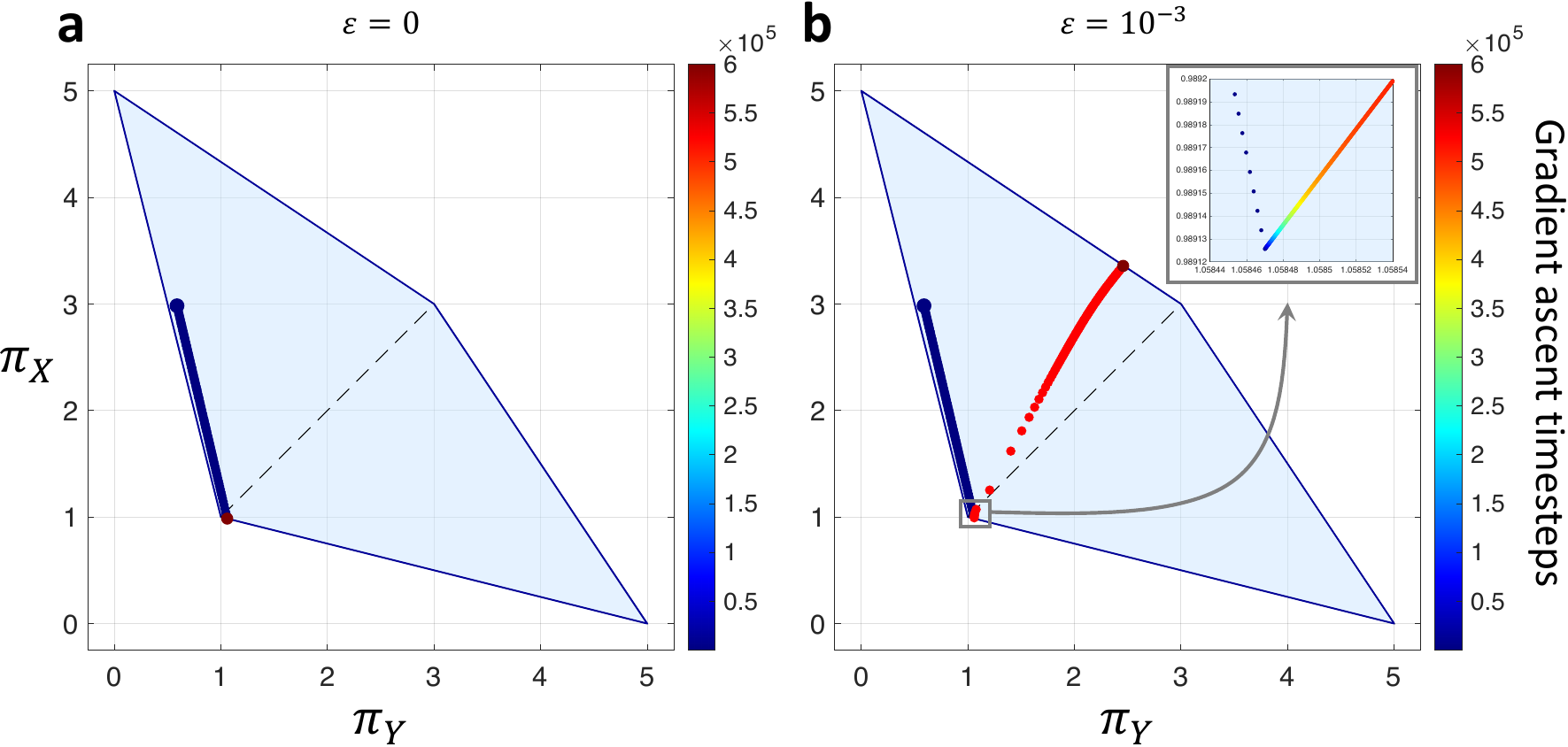}
	\caption{\textbf{Effects of implementation errors on payoff trajectories.} For the (non-ZD) strategy of \fig{distribution_general}\textbf{a}, we know that a significant number of runs lead to a suboptimal outcome. Panel \textbf{a} in the present figure gives an example of such a trajectory. When the optimizing player suffers from implementation errors, he or she plays $\left(1-\varepsilon\right)\mathbf{q}+\varepsilon\left(1-\mathbf{q}\right)$ in place of $\mathbf{q}$, where $\varepsilon$ is the error rate. Introducing a small error rate, $\varepsilon =10^{-3}$ in \textbf{b}, can lead the selfish learner toward the optimal outcome. What is notable here is that it does not necessarily do so efficiently. As \textbf{b} illustrates, the trajectory can still initially lead a learner toward a suboptimal outcome, much in the same way that would be seen without implementation errors. However, in this case the suboptimal outcome can be escaped, although it does take a significant amount of time to do so. Even when errors lead to better outcomes in the long run, selfish optimization can still be quite inefficient.\label{fig:trajectories_with_errors}}
\end{figure}

\section{Discussion}
Understanding the dynamics of strategic interactions between multiple agents is of both theoretical and practical interest. This paper focuses on the two-agent case, in situations modeled by iterated games. Recent literature proves the existence of ``payoff control'' strategies in these situations, which allow a player to unilaterally constrain both players' long-run expected payoffs. The class of zero-determinant (ZD) strategies \citep{pressdyson} enforces a linear relationship between the payoffs, while the strategies described by \citet{payoffcontrol} impose looser constraints, involving one or more linear inequalities among the payoffs. Geometrically, the former strategies restrict the feasible payoff region $\mathcal{C}\left(\mathbf{p}\right)$ to a line, while the latter allow more general control over how the edges of $\mathcal{C}\left(\mathbf{p}\right)$ are oriented. Our study investigates the reliability of these payoff control strategies in practice, when a fixed player $X$, who uses them, faces a selfish co-player $Y$ concerned only with maximizing its own payoff.

\citet{chenzinger} show that in any iterated prisoner's dilemma satisfying $2P<S+T<2R$, all positively correlated ZD (pcZD) strategies are robust even when $Y$ is selfish: all adapting paths of $Y$ lead to the maximum $\pi_{Y}$ without getting stuck at local optima, thereby maximizing $\pi_{X}$ and achieving $X$'s desired outcome. They conclude that ``it is always `safe' for $X$ to use pcZD strategies, and she will receive her desired score in a very robust way, without knowing which adapting path $Y$ will follow.'' There might be a tendency to think that payoff control would be similarly reliable in more general situations as well. For instance, \citet{payoffcontrol} conclude that their (linear-inequality-enforcing) payoff control strategies allow $X$ to ``control the evolutionary route of the game, as long as the opponent is rational and self-optimizing'' and thereby ``enforce the game to finally converge either to a mutual-cooperation equilibrium or to any feasible equilibrium that she wishes.'' Our results demonstrate that controlling the shape of $\mathcal{C}\left(\mathbf{p}\right)$ is not sufficient to control the adapting path of a selfish player, whose objective is to optimize a non-convex function.

We prove analytically that when $Y$ follows a local hill-climbing procedure like projected gradient ascent (PGA), it always reaches the boundary of the hypercube; however, our numerical simulations show that the endpoint is often only locally optimal, with the true optimum lying elsewhere on the boundary. We find this to be the case even when $X$ uses a pcZD strategy in generalized IPDs that do not satisfy Chen and Zinger's constraints. When $X$ uses a general memory-one strategy (the class to which the control strategies described by \citet{payoffcontrol} belong), suboptimal endpoints appear even in the usual IPD satisfying Chen and Zinger's conditions. Furthermore, a significant number of initial strategies can lead to these suboptimal endpoints. The resulting payoffs to both players can, in some cases, be substantially lower than the maximal payoffs. These results show that the efficacy of $X$'s payoff control strategy in achieving its desired outcome is not entirely within $X$'s hands, but partly contingent on $Y$'s initial strategy and learning rule, even when $Y$ is selfishly rational.

While our study focuses on player $X$ using a fixed strategy in a specified game, the implications extend to more general settings. In realistic scenarios, the players may not know which game best describes their interactions, or even whether the game is fixed over the course of their interactions. Their environment could be dynamic and modeled by a stochastic game in which the payoff parameters $\left(R,S,T,P\right)$ change in response to the players' actions \citep{shapley:PNAS:1953,hilbe:Nature:2018}. Although we have not explicitly studied stochastic games here, it seems unlikely that a robustness result would hold in that setting either, at least beyond games involving quite restrictive assumptions. We also note that, although Chen and Zinger's result on the robustness of pcZD strategies in certain IPDs does not extend to general iterated games (even those with static payoff parameters) and strategies, there could indeed be some larger class of game-strategy combinations exhibiting robustness to selfish optimization (in the sense of long-run outcomes). We were not able to readily identify such a pattern. However, due to the ease in finding examples for which selfish optimization leads to multiple endpoints, we would hypothesize that any such class must be quite limited.

One approach that does appear to be beneficial for learning via gradient ascent is the presence of implementation errors. We leave open the question of whether a robustness result for long-run outcomes can be proved in this case for general memory-one strategies. With that said, our examples indicate these benefits might be realized only when the time horizon is sufficiently long. Thus, there are at least two kinds of issues: does gradient ascent lead to optimal outcomes, and, if so, does it do so efficiently? Our findings indicate that, generally speaking, the answer to at least one of these questions must be in the negative. As a result, payoff control and selfish optimization can fail to be effective even in settings of single-agent learning.

Finally, in multi-agent learning, both $X$ and $Y$ might be implementing an optimization process simultaneously. Assuming that $Y$ is selfishly rational, $X$ can try to devise a learning rule with a specific objective in mind (e.g. incentivize $Y$ to repeatedly cooperate). One way of doing so could be to devise a rule that transforms the rightmost point of the feasible region, corresponding to the red dot in \fig{feasible_regions}, in the desired way. Such an approach to multi-agent learning is appealing because it need not take into account what $Y$ is currently doing; rather, what matters is where this rightmost point lies as a function of $X$'s strategy parameters, which essentially allows $X$ to behave as if no opponent is present. But the success of such an approach clearly relies on whether the adapting paths of $Y$ can actually lead to this point, as well as on the amount of time it takes to do so. So, while one might be able to devise such an algorithm, it is not at all clear, based on our results, whether it would be successful in practice against a selfish, myopic opponent. Future research addressing this issue with more sophisticated single-agent techniques would therefore benefit multi-agent methods that rely on the efficacy of selfish optimization.

\setcounter{equation}{0}
\setcounter{figure}{0}
\setcounter{section}{0}
\setcounter{table}{0}
\renewcommand{\thesection}{Appendix~\Alph{section}}
\renewcommand{\thesubsection}{\Alph{section}.\arabic{subsection}}
\renewcommand{\theequation}{\Alph{section}.\arabic{equation}}

\section{Outline and analysis of robustness proof of \citep{chenzinger}}
The robustness result of \citet{chenzinger}, which is summarized in \sect{robustness}, applies to iterated prisoner's dilemmas satisfying $2P<S+T<2R$. In this section, we provide a brief outline of their proof, and highlight why it does not generalize when $S+T>2R$ or $S+T<2P$.

\citet{chenzinger} use game parameters rescaled according to $Z \rightarrow \left(Z-P\right) /\left(R-P\right)$, so that it may be assumed that $R=1$ and $P=0$ without a loss of generality. With these rescaled parameters, all prisoner's dilemmas satisfy $T>1$ and $S<0$, and $\kappa\in\left[0,1\right]$ (where $\kappa$ is the parameter in the ZD-imposed relation $\pi_{X}-\kappa =\chi\left(\pi_{Y}-\kappa\right)$). Following the terminology of \citet{chenzinger}, we use the term ``positively correlated zero-determinant (pcZD)'' strategies to refer to ZD strategies with positive $\chi$. The fixed player, $X$, uses a pcZD strategy $\mathbf{p}$, while the (selfish) adapting player, $Y$, uses $\mathbf{q}$. The key result is Proposition 2 of \citet{chenzinger}, which states that when $X$ uses a pcZD strategy, $\frac{\partial\pi_{Y}\left(\mathbf{p},\mathbf{q}\right)}{\partial q_{xy}}\geqslant 0$ for each $x,y\in\left\{C,D\right\}$, at all strategies $\mathbf{q}=\left(q_{CC},q_{CD},q_{DC},q_{DD}\right)$ where the Markov limiting distribution is well-defined. After considering the edge cases where any of these derivatives vanish (an essential step for the argument to apply to all adapting paths), this result can be used to show that $Y$'s optimized strategy along any adapting path is repeated cooperation.

To prove the above condition on $\pi_{Y}$, it suffices to prove the same condition on $\pi_{X}$, since their derivatives are related by the positive proportionality factor $\chi$ when $\mathbf{p}$ is a pcZD strategy. \citet{chenzinger} show that, for $x,y\in\left\{C,D\right\}$, the derivative of $\pi_{X}$ with respect to $q_{xy}$ satisfies
\begin{align}
D\left(\mathbf{p},\mathbf{q}\right)^{2} \frac{\partial\pi_{X}\left(\mathbf{p},\mathbf{q}\right)}{\partial q_{xy}} = \left( 1-p_{CD}-\left(1-p_{CC}\right) S+p_{DD}\left(1-S\right)\right) A_{xy} B_{xy} , \label{eq:partial_piX}
\end{align}
where \emph{(i)} $D\left(\mathbf{p},\mathbf{q}\right)$ is a quantity that is nonzero whenever the Markov chain defined by \eq{markov_matrix} has a unique stationary distribution and \emph{(ii)} the terms $A_{xy}$ and $B_{xy}$ have no dependence on $q_{xy}$ and are linear in $q_{x'y'}$ for $\left(x',y'\right)\neq\left(x,y\right)$. (We slightly modify their notation for clarity.) By \emph{(i)}, the sign of each derivative term is determined by the three factors on the right-hand side of \eq{partial_piX} whenever \eq{markov_matrix} has a unique stationary distribution. By \emph{(ii)}, the sign of $A_{xy}B_{xy}$ is fully determined by its signs at the vertices of the cube $\left[0,1\right]^{3}$. Using the fact that $p_{CC},p_{CD},p_{DC},p_{DD}\in\left[0,1\right]$, one sees by inspection that $-A_{CC}$, $A_{CD}$, $-A_{DC}$, and $A_{DD}$ are always non-negative at every vertex. Also, by inspection, the terms $-B_{CC}$, $B_{CD}$, $-B_{DC}$, and $B_{DD}$ are non-negative at every vertex if the three conditions
\begin{subequations}\label{eq:B_conditions}
\begin{align}
p_{CC} \geqslant p_{CD} ; \\
p_{DC} \geqslant p_{DD} ; \\
2 \geqslant S+T \geqslant 0
\end{align}
\end{subequations}
are satisfied. These three conditions, together with a fourth condition,
\begin{align}
1-p_{CD}-\left(1-p_{CC}\right) S+p_{DD}\left(1-S\right) \geqslant 0 , \label{eq:AB_prefactor}
\end{align}
would imply that the right-hand side of \eq{partial_piX} is non-negative, thereby giving $\frac{\partial\pi_{Y}\left(\mathbf{p},\mathbf{q}\right)}{\partial q_{xy}}\geqslant 0$.

We now examine these four conditions in iterated prisoner's dilemmas, assuming throughout that both $\chi$ and $\phi$ are positive since we are concerned with pcZD strategies. Using \eq{ZD_formula} for a generic zero-determinant strategy, one can show that \eq{B_conditions}\textbf{a} and \eq{B_conditions}\textbf{b} are equivalent to 
\begin{subequations}\label{eq:B_conditions_ZD}
\begin{align}
\chi\left(T-1\right) +\left(1-S\right) &\geqslant 0 ; \\
T-\chi S &\geqslant 0 ,
\end{align}
\end{subequations}
respectively, for any pcZD strategy $\mathbf{p}$.

All prisoner's dilemma games satisfy \eq{AB_prefactor} and \eq{B_conditions_ZD} with strict inequalities, due to the ranking $T>R=1>P=0>S$. The remaining condition is simply the rescaled version of $2P \leqslant S+T \leqslant 2R$. Therefore, IPDs that satisfy Chen and Zinger's constraint $2P<S+T<2R$ meet all four conditions required for $\frac{\partial\pi_{Y}\left(\mathbf{p},\mathbf{q}\right)}{\partial q_{xy}}\geqslant 0$ to hold. As detailed by \citet{chenzinger}, this implies (after consideration of vanishing derivatives) that the endpoint of any adapting path must satisfy $q_{CC}=q_{CD}=1$ or $p_{CC}=q_{CC}=1$, both of which correspond to the payoff-maximizing behavior of repeated cooperation by $Y$ (in the latter case, the Markov limiting distribution is repeated cooperation by both players). However, IPDs with $S+T>2R$ or $S+T<2P$ violate \eq{B_conditions}\textbf{c}, allowing for the partial derivatives $\frac{\partial\pi_{Y}\left(\mathbf{p},\mathbf{q}\right)}{\partial q_{xy}}$ to be negative. For instance, with game parameters $(2,-1,7,0)$ as in \fig{GA_trajectories}, pcZD strategy $\mathbf{p} = (1,0.12,0.88,0)$, and $\mathbf{q} = (0.08,0.77,0.95,0.68)$, the partial derivatives $\frac{\partial\pi_{Y}\left(\mathbf{p},\mathbf{q}\right)}{\partial q_{CC}}$ and $\frac{\partial\pi_{Y}\left(\mathbf{p},\mathbf{q}\right)}{\partial q_{CD}}$ are negative. This creates the possibility that strategies other than repeated cooperation can be local optima and thus stable endpoints of adapting paths. Our simulations, discussed in \sect{methods_and_results}, show that this possibility can occur in practice.

\section{Equalizers and stable points of gradient ascent}
In a repeated game with game parameters $R$, $S$, $T$, and $P$, \citet{pressdyson} established an explicit formula for payoffs when $X$ uses $\mathbf{p}$ and $Y$ uses $\mathbf{q}$. Specifically, $Y$'s payoff is given by
\begin{align}
\pi_{Y}\left(\mathbf{p},\mathbf{q}\right) &= \frac{\det\begin{pmatrix}
		p_{CC}q_{CC}-1 & p_{CC}-1 & q_{CC}-1 & R \\
		p_{DC}q_{CD} & p_{DC} & q_{CD}-1 & S \\
		p_{CD}q_{DC} & p_{CD}-1 & q_{DC} & T \\
		p_{DD}q_{DD} & p_{DD} & q_{DD} & P \\
\end{pmatrix}}{\det\begin{pmatrix}
p_{CC}q_{CC}-1 & p_{CC}-1 & q_{CC}-1 & 1 \\
p_{DC}q_{CD} & p_{DC} & q_{CD}-1 & 1 \\
p_{CD}q_{DC} & p_{CD}-1 & q_{DC} & 1 \\
p_{DD}q_{DD} & p_{DD} & q_{DD} & 1 \\
\end{pmatrix}} .
\end{align}
By linearity of the determinant, there are multilinear functions $f_{CC}$, $f_{DC}$, $f_{CD}$, and $f_{DD}$ such that
\begin{align}
\pi_{Y}\left(\mathbf{p},\mathbf{q}\right) &= \frac{f_{CC}\left(\mathbf{p},\mathbf{q}\right) R+f_{DC}\left(\mathbf{p},\mathbf{q}\right) S+f_{CD}\left(\mathbf{p},\mathbf{q}\right) T+f_{DD}\left(\mathbf{p},\mathbf{q}\right) P}{f_{CC}\left(\mathbf{p},\mathbf{q}\right) +f_{DC}\left(\mathbf{p},\mathbf{q}\right) +f_{CD}\left(\mathbf{p},\mathbf{q}\right) +f_{DD}\left(\mathbf{p},\mathbf{q}\right)} . \label{eq:piY_f}
\end{align}
In what follows, we let $f_{\Sigma}\coloneqq f_{CC}+f_{DC}+f_{CD}+f_{DD}$.

\begin{equalizertheorem}
If $\mathbf{p}\in\left[0,1\right]^{4}\setminus\left\{\left(1,1,0,0\right)\right\}$, then $V\left(\mathbf{p},\mathbf{q}\right) =E\left(\mathbf{p}\right)$ whenever $\mathbf{q}\in\left(0,1\right)^{4}$.
\end{equalizertheorem}
\begin{proof}
Let ``$\ker$'' denote the right null space of a matrix. Using \eq{piY_f}, together with the fact that $\partial f_{CC}/\partial q_{CC}=\partial f_{DC}/\partial q_{CD}=\partial f_{CD}/\partial q_{DC}=\partial f_{DD}/\partial q_{DD}=0$, we can express \eq{V_def} as
\begin{align}
V\left(\mathbf{p},\mathbf{q}\right) &= \ker\begin{pmatrix}
-\frac{f_{CC}}{f_{\Sigma}}\frac{\partial f_{\Sigma}}{\partial q_{CC}} & \frac{\partial f_{DC}}{\partial q_{CC}}-\frac{f_{DC}}{f_{\Sigma}}\frac{\partial f_{\Sigma}}{\partial q_{CC}} & \frac{\partial f_{CD}}{\partial q_{CC}}-\frac{f_{CD}}{f_{\Sigma}}\frac{\partial f_{\Sigma}}{\partial q_{CC}} & \frac{\partial f_{DD}}{\partial q_{CC}}-\frac{f_{DD}}{f_{\Sigma}}\frac{\partial f_{\Sigma}}{\partial q_{CC}} \\
\frac{\partial f_{CC}}{\partial q_{CD}}-\frac{f_{CC}}{f_{\Sigma}}\frac{\partial f_{\Sigma}}{\partial q_{CD}} & -\frac{f_{DC}}{f_{\Sigma}}\frac{\partial f_{\Sigma}}{\partial q_{CD}} & \frac{\partial f_{CD}}{\partial q_{CD}}-\frac{f_{CD}}{f_{\Sigma}}\frac{\partial f_{\Sigma}}{\partial q_{CD}} & \frac{\partial f_{DD}}{\partial q_{CD}}-\frac{f_{DD}}{f_{\Sigma}}\frac{\partial f_{\Sigma}}{\partial q_{CD}} \\
\frac{\partial f_{CC}}{\partial q_{DC}}-\frac{f_{CC}}{f_{\Sigma}}\frac{\partial f_{\Sigma}}{\partial q_{DC}} & \frac{\partial f_{DC}}{\partial q_{DC}}-\frac{f_{DC}}{f_{\Sigma}}\frac{\partial f_{\Sigma}}{\partial q_{DC}} & -\frac{f_{CD}}{f_{\Sigma}}\frac{\partial f_{\Sigma}}{\partial q_{DC}} & \frac{\partial f_{DD}}{\partial q_{DC}}-\frac{f_{DD}}{f_{\Sigma}}\frac{\partial f_{\Sigma}}{\partial q_{DC}} \\
\frac{\partial f_{CC}}{\partial q_{DD}}-\frac{f_{CC}}{f_{\Sigma}}\frac{\partial f_{\Sigma}}{\partial q_{DD}} & \frac{\partial f_{DC}}{\partial q_{DD}}-\frac{f_{DC}}{f_{\Sigma}}\frac{\partial f_{\Sigma}}{\partial q_{DD}} & \frac{\partial f_{CD}}{\partial q_{DD}}-\frac{f_{CD}}{f_{\Sigma}}\frac{\partial f_{\Sigma}}{\partial q_{DD}} & -\frac{f_{DD}}{f_{\Sigma}}\frac{\partial f_{\Sigma}}{\partial q_{DD}}
\end{pmatrix} .
\end{align}
Let $M^{V}$ denote the matrix in this equation. A straightforward calculation shows that the determinants of all $3\times 3$ minors of $M^{V}$ vanish, which means that $\dim V\left(\mathbf{p},\mathbf{q}\right)\geqslant 2$. Furthermore, since $\mathbf{p}\neq\left(1,1,0,0\right)$ and $\mathbf{q}\in\left(0,1\right)^{4}$, one can check that all determinants of $2\times 2$ minors of $M^{V}$ vanish if and only if one of the following holds: \emph{(i)} $p_{CC}=p_{CD}=1$, \emph{(ii)} $p_{DC}=p_{DD}=0$, or \emph{(iii)} $p_{CC}=p_{CD}=p_{DC}=p_{DD}$. The only strategy that satisfies one of these conditions \emph{and} causes $M^{V}$ to vanish completely is $\mathbf{p}=\left(1,1,0,0\right)$, which is excluded. Therefore, $\dim V\left(\mathbf{p},\mathbf{q}\right)$ is generically $2$, deviating (to $3$) only if $p_{CC}=p_{CD}=1$, $p_{DC}=p_{DD}=0$, or $p_{CC}=p_{CD}=p_{DC}=p_{DD}$.

Clearly, $E\left(\mathbf{p}\right)\subseteq V\left(\mathbf{p},\mathbf{q}\right)$ by the definition of an equalizer strategy. However, we need a more tangible characterization of $E\left(\mathbf{p}\right)$ in order to argue that $E\left(\mathbf{p}\right) =V\left(\mathbf{p},\mathbf{q}\right)$. The idea is to consider the boundary strategies with purely deterministic components, but we also include a buffer parameter, $\varepsilon\in\left(0,1/2\right)$, to ensure that payoffs are always defined. For $\varepsilon\in\left(0,1/2\right)$, let
\begin{align}
B_{\varepsilon}\left(\mathbf{p}\right) &\coloneqq \left\{\begin{pmatrix}R \\ S \\ T \\ P\end{pmatrix}\in\mathbb{R}^{4}\ :\ \pi_{Y}\left(\mathbf{p},\mathbf{q'}\right)\textrm{ is independent of }\mathbf{q'}\in\left\{\varepsilon ,1-\varepsilon\right\}^{4}\right\} .
\end{align}
Again, $E\left(\mathbf{p}\right)\subseteq B_{\varepsilon}\left(\mathbf{p}\right)$ for every $\varepsilon\in\left(0,1/2\right)$ by the definition of an equalizer strategy. Conversely, if $\mathbf{p}$ is not an equalizer strategy, then we claim that there must exist $\varepsilon =\varepsilon\left(\mathbf{p}\right)$ with $0<\varepsilon\ll 1$ such that $\pi_{Y}\left(\mathbf{p},\mathbf{q'}\right)$ is \emph{not} independent of $\mathbf{q'}\in\left\{\varepsilon ,1-\varepsilon\right\}^{4}$. For non-equalizer $\mathbf{p}$, there must exist $\mathbf{q}^{1},\mathbf{q}^{2}\in\left[0,1\right]^{4}$ such that $\pi_{Y}\left(\mathbf{p},\mathbf{q}^{1}\right) >\pi_{Y}\left(\mathbf{p},\mathbf{q}^{2}\right)$. By continuity, we may assume that $\mathbf{q}^{1},\mathbf{q}^{2}\in\left(0,1\right)^{4}$. (Note that the values of $\mathbf{q}^{1}$ and $\mathbf{q}^{2}$ in $\left(0,1\right)^{4}$ may be considered functions of $\mathbf{p}$, which is how we will derive $\varepsilon$ as a function of $\mathbf{p}$.) Holding all other strategy components fixed, we know that for each $x,y\in\left\{C,D\right\}$, $\pi_{Y}$ is either independent of $q_{xy}$ or strictly monotonic in $q_{xy}$. (This can be seen by examining \eq{piY_f}; see also \citep{mcavoy}.) For $0<\varepsilon <\min_{x,y\in\left\{C,D\right\}}\left\{q_{xy}^{1},1-q_{xy}^{1},q_{xy}^{2},1-q_{xy}^{2}\right\}$, we define two new strategies $\left(\mathbf{q}^{1}\right) ',\left(\mathbf{q}^{2}\right) '\in\left\{\varepsilon ,1-\varepsilon\right\}^{4}$ as follows: Let $\left(\mathbf{q};q_{xy}=r\right)$ be the strategy obtained by taking $\mathbf{q}$ and replacing $q_{xy}$ by $r$. Start with $\left(\mathbf{q}^{1}\right) '=\mathbf{q}^{1}$, and for $\left(x,y\right)\in\left\{\left(C,C\right) ,\left(C,D\right) ,\left(D,C\right) ,\left(D,D\right)\right\}$ (in that order), we sequentially modify $\left(\mathbf{q}^{1}\right) '$ to give a non-decreasing sequence of payoffs to $Y$ by letting
\begin{align}
\left(\mathbf{q}^{1}\right)_{xy} &= 
\begin{cases}
1-\varepsilon & \pi_{Y}\left(\mathbf{p},\left(\left(\mathbf{q}^{1}\right) ';\left(q^{1}\right)_{xy}'=1-\varepsilon\right)\right)\geqslant\pi_{Y}\left(\mathbf{p},\left(\mathbf{q}^{1}\right) '\right) , \\
& \\
\varepsilon & \pi_{Y}\left(\mathbf{p},\left(\left(\mathbf{q}^{1}\right) ';\left(q^{1}\right)_{xy}'=1-\varepsilon\right)\right) <\pi_{Y}\left(\mathbf{p},\left(\mathbf{q}^{1}\right) '\right) .
\end{cases}
\end{align}
Similarly, for $\left(\mathbf{q}^{2}\right) '$, we start with $\left(\mathbf{q}^{2}\right) '=\mathbf{q}^{2}$ and make the successive changes
\begin{align}
\left(\mathbf{q}^{2}\right)_{xy} &= 
\begin{cases}
1-\varepsilon & \pi_{Y}\left(\mathbf{p},\left(\left(\mathbf{q}^{2}\right) ';\left(q^{2}\right)_{xy}'=1-\varepsilon\right)\right) <\pi_{Y}\left(\mathbf{p},\left(\mathbf{q}^{2}\right) '\right) , \\
& \\
\varepsilon & \pi_{Y}\left(\mathbf{p},\left(\left(\mathbf{q}^{2}\right) ';\left(q^{2}\right)_{xy}'=1-\varepsilon\right)\right)\geqslant\pi_{Y}\left(\mathbf{p},\left(\mathbf{q}^{2}\right) '\right) ,
\end{cases}
\end{align}
which gives a non-increasing sequence of payoffs to $Y$. After all of these updates, we have
\begin{align}
\pi_{Y}\left(\mathbf{p},\left(\mathbf{q}^{1}\right) '\right) \geqslant \pi_{Y}\left(\mathbf{p},\mathbf{q}^{1}\right) > \pi_{Y}\left(\mathbf{p},\mathbf{q}^{2}\right) \geqslant \pi_{Y}\left(\mathbf{p},\left(\mathbf{q}^{2}\right) '\right) .
\end{align}
Therefore, we see that $B_{\varepsilon}\left(\mathbf{p}\right)\subseteq E\left(\mathbf{p}\right)$, and thus that $B_{\varepsilon}\left(\mathbf{p}\right) = E\left(\mathbf{p}\right) \subseteq V\left(\mathbf{p},\mathbf{q}\right)$.

Just as $V$ is given by the null space of an explicit matrix ($M^{V}$), so too is $B_{\varepsilon}$. For $\mathbf{q}\in\left\{\varepsilon ,1-\varepsilon\right\}^{4}\setminus\left\{\left(\varepsilon ,\varepsilon ,\varepsilon ,\varepsilon\right)\right\}$, we let the row of $M^{B_{\varepsilon}}$ corresponding to $\mathbf{q}$ be the unique vector $v$ with
\begin{align}
\pi_{Y}\left(\mathbf{p},\mathbf{q}\right) - \pi_{Y}\left(\mathbf{p},\left(\varepsilon ,\varepsilon ,\varepsilon ,\varepsilon\right)\right) &= \left< v, \left(R,S,T,P\right) \right> .
\end{align}
$M^{B_{\varepsilon}}$ is thus a $15\times 4$ matrix, depending on $\mathbf{p}$, for which $B_{\varepsilon}\left(\mathbf{p}\right) =\ker M^{B_{\varepsilon}}$. A direct calculation shows that all determinants of $3\times 3$ minors of $M^{B_{\varepsilon}}$ vanish, regardless of $\varepsilon$. Thus, $\dim B_{\varepsilon}\left(\mathbf{p}\right)\geqslant 2$. Similarly, when $p_{CC}=p_{CD}=1$, $p_{DC}=p_{DD}=0$, or $p_{CC}=p_{CD}=p_{DC}=p_{DD}$, we see that all determinants of $2\times 2$ minors of $M^{B_{\varepsilon}}$ vanish (again, regardless of $\varepsilon$). By our remarks about $V\left(\mathbf{p},\mathbf{q}\right)$, the latter implies that $B_{\varepsilon}\left(\mathbf{p}\right)$ has dimension at least $3$ if and only if $V\left(\mathbf{p},\mathbf{q}\right)$ has dimension exactly $3$. Since $B_{\varepsilon}\left(\mathbf{p}\right) = E\left(\mathbf{p}\right) \subseteq V\left(\mathbf{p},\mathbf{q}\right)$, it follows that $E\left(\mathbf{p}\right) = V\left(\mathbf{p},\mathbf{q}\right)$.
\end{proof}

\begin{remark}
The proof of Theorem~\ref{thm:equalizer} is supplemented by the MATLAB file theorem1.m.
\end{remark}

\begin{corollary}
For any $\mathbf{p}\in\left[0,1\right]^{4}\setminus\left\{\left(1,1,0,0\right)\right\}$, the dimension of $E\left(\mathbf{p}\right)$ is $2$ unless $p_{CC}=p_{CD}=1$, $p_{DC}=p_{DD}=0$, or $p_{CC}=p_{CD}=p_{DC}=p_{DD}$, in which case $\dim E\left(\mathbf{p}\right) =3$.
\end{corollary}

\begin{remark}
When discussing equalizer strategies, \citet{pressdyson} show that if there exist $\beta ,\gamma\in\mathbb{R}$ such that $\mathbf{p}=\left(1+\beta R+\gamma ,1+\beta T+\gamma ,\beta S+\gamma ,\beta P+\gamma\right)$, then by using $\mathbf{p}$, player $X$ can ensure $\beta\pi_{Y}+\gamma =0$, regardless of the strategy of $Y$. (Note that this is the same condition given by \citet{boerlijst:AMM:1997}, who discovered equalizer strategies well before the more general class of zero-determinant strategies was discovered by \citet{pressdyson}.) In particular, if $\beta\neq 0$, then $X$ can enforce the equalizer condition $\pi_{Y}=-\gamma /\beta$. Based on this condition, we define
\begin{align}
E_{0}\left(\mathbf{p}\right) &\coloneqq \left\{\begin{pmatrix}R \\ S \\ T \\ P\end{pmatrix}\in\mathbb{R}^{4}\ :\ \mathbf{p}=\begin{pmatrix}1+\beta R+\gamma \\ 1+\beta T+\gamma \\ \beta S+\gamma \\ \beta P+\gamma\end{pmatrix}\textrm{ for some }\beta\neq 0\textrm{ and }\gamma\in\mathbb{R}\right\} .
\end{align}
By the work of \citet{boerlijst:AMM:1997} and \citet{pressdyson}, we have $E_{0}\left(\mathbf{p}\right)\subseteq E\left(\mathbf{p}\right)$. It is natural to ask whether this inclusion can ever be strict. For a strategy of the form $\mathbf{p}=\left(p_{CC},1,0,0\right)$ with $p_{CC}\neq 1$, we have $\left(R,S,T,P\right)\in E\left(\mathbf{p}\right)$ if $S=P$. In this case, for $\left(R,S,T,S\right)\in E_{0}\left(\mathbf{p}\right)$, there must exist $\beta\neq 0$ and $\gamma\in\mathbb{R}$ such that $p_{CC}=1+\beta R+\gamma$ and $\beta S+\gamma =\beta T+\gamma =0$. The latter equation can hold only when $S=T$. Thus, when $S\neq T$, we have $E_{0}\left(\mathbf{p}\right)\subsetneq E\left(\mathbf{p}\right)$ for this $\mathbf{p}$.
\end{remark}

Finally, we turn to the proof of Proposition~\ref{prop:endpoints}, our result on the endpoints of gradient ascent:
\begin{propendpoints}
Suppose $X$ and $Y$ are general memory-one players in any iterated symmetric two-player two-action game, $\mathbf{p}$ is a fixed strategy of $X$, and $\mathbf{q}_{\text{final}}$ is $Y$'s optimized strategy following a projected gradient trajectory starting from any initial strategy, $\mathbf{q}_{0}$. Unless $\mathbf{p}$ is an equalizer strategy (the degenerate case corresponding to a flat payoff landscape), $\mathbf{q}_{\text{final}}$ is constrained as follows:
\begin{enumerate}

\item[(a)] For all $\mathbf{p}\in\left[0,1\right]^{4}\setminus\left\{\left(1,1,0,0\right)\right\}$, $\mathbf{q}_{\text{final}}$ has at least one deterministic component;

\item[(b)] If $\mathbf{p}$ is randomly sampled from $\left[0,1\right]^{4}$ (with each of the four components independent), then with probability one, $\mathbf{q}_{\text{final}}$ is one of the following: $\left(1,1,q_{DC},q_{DD}\right)$, $\left(q_{CC},q_{CD},0,0\right)$, or a strategy with all four components deterministic.

\end{enumerate}
\end{propendpoints}
\begin{proof}
Part \emph{(a)} is an immediate corollary of \thm{equalizer}. To prove part \emph{(b)}, we assume that the components of $\mathbf{p}$ are sampled independently and simultaneously from an arcsine distribution. The argument is not specific to this distribution, however, and can be easily adapted to other kinds of sampling methods (e.g. uniform). First, for a subset $W\subseteq\left\{CC,CD,DC,DD\right\}$, we let
\begin{align}
V_{W}\left(\mathbf{p},\mathbf{q}\right) &\coloneqq \left\{\begin{pmatrix}R \\ S \\ T \\ P\end{pmatrix}\in\mathbb{R}^{4}\ :\ \frac{\partial\pi_{Y}\left(\mathbf{p},\mathbf{q}\right)}{\partial q_{xy}}=0 \textrm{ for every }xy\in W\right\} .
\end{align}
$V_{W}$ is a generalization of the space $V$ of \eq{V_def}, which satisfies $V\left(\mathbf{p},\mathbf{q}\right) =V_{\left\{CC,CD,DC,DD\right\}}\left(\mathbf{p},\mathbf{q}\right)$ and $V_{W}\left(\mathbf{p},\mathbf{q}\right)\subseteq V_{W'}\left(\mathbf{p},\mathbf{q}\right)$ whenever $W'\subseteq W$. By the argument we used in the proof of Theorem~\ref{thm:equalizer}, we know that $E\left(\mathbf{p}\right)$ has dimension $2$ with probability one since it deviates to $3$ only when $p_{CC}=p_{CD}=1$, $p_{DC}=p_{DD}=0$, or $p_{CC}=p_{CD}=p_{DC}=p_{DD}$ (each of which occurs with probability zero). Notably, in the proof of Theorem~\ref{thm:equalizer}, we saw that for $\mathbf{p}\in\left[0,1\right]^{4}\setminus\left\{\left(1,1,0,0\right)\right\}$ and $\mathbf{q}\in\left(0,1\right)^{4}$, the dimension of $V\left(\mathbf{p},\mathbf{q}\right)$ is determined by $\mathbf{p}$ alone. This property does not extend to all $\mathbf{q}$ on the boundary, but it does hold in many cases, so in the following we use the reasoning of the proof of Theorem~\ref{thm:equalizer} and handle the exceptional cases as they arise.

The case of $W\subseteq\left\{CC,CD,DC,DD\right\}$ with $\left| W\right| =3$ is straightforward. Suppose that $q_{xy}\in\left(0,1\right)$ when $xy\in W$ and $q_{xy}\in\left\{0,1\right\}$ when $xy\not\in W$. We know that $\dim V_{W}\left(\mathbf{p},\mathbf{q}\right)\geqslant 2$, and one can reduce the condition $\dim V_{W}\left(\mathbf{p},\mathbf{q}\right) >2$ to a system of equations in $\mathbf{p}$. Importantly, these equations are not all identically zero, so the probability that they are all satisfied when $\mathbf{p}$ is chosen randomly is zero. As a result, we have $E\left(\mathbf{p}\right) =V_{W}\left(\mathbf{p},\mathbf{q}\right)$ like in the proof of Theorem~\ref{thm:equalizer}.

The case of $\left| W\right| =2$, with $q_{xy}\in\left(0,1\right)$ when $xy\in W$ and $q_{xy}\in\left\{0,1\right\}$ when $xy\not\in W$, is similar. However, there are two exceptional cases: $W=\left\{DC,DD\right\}$ and $q_{CC}=q_{CD}=1$ and $W=\left\{CC,CD\right\}$ and $q_{DC}=q_{DD}=0$. In each of these two cases, the matrix $M^{V_{W}}$ itself is identically zero (where $M^{V_{W}}$ is the matrix such that $\ker M_{V_{W}}=V_{W}$), so clearly strategies of the form $\mathbf{q}=\left(1,1,q_{DC},q_{DD}\right)$ and $\mathbf{q}=\left(q_{CC},q_{CD},0,0\right)$ are stable points of gradient ascent. In each of the $22$ other cases, straightforward calculations show that the equations in $\mathbf{p}$ that need to be satisfied for $\dim V_{W}\left(\mathbf{p},\mathbf{q}\right) >2$ to hold are not all identically zero, and thus with probability one we must have $E\left(\mathbf{p}\right) =V_{W}\left(\mathbf{p},\mathbf{q}\right)$ based on the reasoning used in the proof of Theorem~\ref{thm:equalizer}.

Finally, we note that if $\left| W\right| =1$, then $\dim V_{W}\left(\mathbf{p},\mathbf{q}\right)\geqslant 3$, so this dimensionality argument does not work. What this means is that for any $\mathbf{p}$, one can find game parameters $\left(R,S,T,P\right)$ for which $\mathbf{p}$ is not an equalizer, yet $Y$ has a strategy $\mathbf{q}$ that is a stable point for gradient ascent with three components in $\left\{0,1\right\}$ and one in $\left(0,1\right)$. This does not mean that \emph{(b)} is false, however, because it is a statement about sampling $\mathbf{p}$ in a game with \emph{fixed} parameters, $\left(R,S,T,P\right)$. Due to the fact that we are now considering just one vanishing partial derivative ($\left| W\right| =1$), we can employ a slightly different line of reasoning. We illustrate the argument for $W=\left\{CC\right\}$, with the other three cases being similar. Since whether $\frac{\partial\pi_{Y}\left(\mathbf{p},\mathbf{q}\right)}{\partial q_{CC}}=0$ or not is independent of $q_{CC}$, and since $\left(q_{CD},q_{DC},q_{DD}\right)\in\left\{0,1\right\}^{3}$, it reduces to a (nontrivial) multilinear constraint on $\mathbf{p}$. The probability that this constraint is satisfied is zero under the sampling scheme for $\mathbf{p}$, which means that $\frac{\partial\pi_{Y}\left(\mathbf{p},\mathbf{q}\right)}{\partial q_{CC}}\neq 0$ almost everywhere. This completes the proof of part \emph{(b)} of the proposition.
\end{proof}

\begin{remark}
The proof of Proposition~\ref{prop:endpoints} is supplemented by the MATLAB file proposition1.m.
\end{remark}

\section*{Acknowledgments}
We thank Martin Nowak for helpful discussions and for suggesting the exploration of implementation errors. This work was supported by the Army Research Laboratory (grant W911NF-18-2-0265) and the Simons Foundation (Math+X Grant to the University of Pennsylvania). Part of this work was carried out while Arjun Mirani was supported by the Harvard College Research Program (HCRP) and the Harvard University Department of Physics.

\section*{Code availability}
Supporting code is available at https://github.com/alexmcavoy/selfish-robustness/.


\begin{thebibliography}{30}
	\providecommand{\natexlab}[1]{#1}
	\providecommand{\url}[1]{\texttt{#1}}
	\expandafter\ifx\csname urlstyle\endcsname\relax
	\providecommand{\doi}[1]{doi: #1}\else
	\providecommand{\doi}{doi: \begingroup \urlstyle{rm}\Url}\fi
	
	\bibitem[Press and Dyson(2012)]{pressdyson}
	W.~H. Press and F.~J. Dyson.
	\newblock Iterated prisoner's dilemma contains strategies that dominate any
	evolutionary opponent.
	\newblock \emph{Proceedings of the National Academy of Sciences}, 109\penalty0
	(26):\penalty0 10409--10413, 2012.
	\newblock \doi{10.1073/pnas.1206569109}.
	
	\bibitem[Chen and Zinger(2014)]{chenzinger}
	J.~Chen and A.~Zinger.
	\newblock The robustness of zero-determinant strategies in iterated prisoner's
	dilemma games.
	\newblock \emph{Journal of Theoretical Biology}, 357:\penalty0 46--54, 2014.
	\newblock \doi{10.1016/j.jtbi.2014.05.004}.
	
	\bibitem[Axelrod(1984)]{axelrodbook}
	R.~Axelrod.
	\newblock \emph{The Evolution of Cooperation}.
	\newblock Basic, New York, 1984.
	
	\bibitem[Rapoport(1965)]{rapoport}
	A.~Rapoport.
	\newblock \emph{Prisoner's Dilemma; A Study in Conflict and Cooperation}.
	\newblock Human Relations Collection. University of Michigan Press, Ann Arbor,
	1965.
	
	\bibitem[Maynard~Smith(1982)]{maynard}
	J.~Maynard~Smith.
	\newblock \emph{Evolution and the Theory of Games}.
	\newblock Cambridge University Press, Cambridge, 1982.
	
	\bibitem[Hofbauer and Sigmund(1998)]{sigmundbook}
	J.~Hofbauer and K.~Sigmund.
	\newblock \emph{Evolutionary Games and Population Dynamics}.
	\newblock Cambridge University Press, West Nyack, 1998.
	
	\bibitem[Nowak(2006{\natexlab{a}})]{nowak:fiverules}
	M.~A. Nowak.
	\newblock {Five Rules for the Evolution of Cooperation}.
	\newblock \emph{Science}, 314\penalty0 (5805):\penalty0 1560--1563,
	2006{\natexlab{a}}.
	\newblock \doi{10.1126/science.1133755}.
	
	\bibitem[Nowak(2006{\natexlab{b}})]{nowak:book}
	M.~A. Nowak.
	\newblock \emph{Evolutionary Dynamics: {Exploring} the Equations of Life}.
	\newblock The Belknap Press of Harvard University Press, 2006{\natexlab{b}}.
	
	\bibitem[Doebeli and Hauert(2005)]{doebeli:ecology:2005}
	M.~Doebeli and C.~Hauert.
	\newblock {Models of cooperation based on the Prisoner's Dilemma and the
		Snowdrift game}.
	\newblock \emph{Ecology Letters}, 8\penalty0 (7):\penalty0 748--766, 2005.
	\newblock \doi{10.1111/j.1461-0248.2005.00773.x}.
	
	\bibitem[Imhof et~al.(2005)Imhof, Fudenberg, Nowak, and May]{lorens:PNAS:2005}
	L.~A. Imhof, D.~Fudenberg, M.~A. Nowak, and R.~M. May.
	\newblock {Evolutionary Cycles of Cooperation and Defection}.
	\newblock \emph{Proceedings of the National Academy of Sciences}, 102\penalty0
	(31):\penalty0 10797--10800, 2005.
	\newblock \doi{10.1073/pnas.0502589102}.
	
	\bibitem[Van~Segbroeck et~al.(2012)Van~Segbroeck, Pacheco, Lenaerts, and
	Santos]{segbroeck:PRL:2012}
	S.~Van~Segbroeck, J.~M. Pacheco, T.~Lenaerts, and F.~C. Santos.
	\newblock {Emergence of Fairness in Repeated Group Interactions}.
	\newblock \emph{Physical Review Letters}, 108\penalty0 (15):\penalty0
	158104--158104, 2012.
	\newblock \doi{10.1103/PhysRevLett.108.158104}.
	
	\bibitem[Nash(1950)]{nash}
	J.~F. Nash.
	\newblock Equilibrium points in n-person games.
	\newblock \emph{Proceedings of the National Academy of Sciences}, 36\penalty0
	(1):\penalty0 48--49, 1950.
	\newblock \doi{10.1073/pnas.36.1.48}.
	
	\bibitem[Rand and Nowak(2013)]{rand:CogSci:2013}
	D.~G. Rand and M.~A. Nowak.
	\newblock Human cooperation.
	\newblock \emph{Trends in Cognitive Sciences}, 17\penalty0 (8):\penalty0
	413--425, 2013.
	\newblock \doi{10.1016/j.tics.2013.06.003}.
	
	\bibitem[Baek et~al.(2016)Baek, Jeong, Hilbe, and Nowak]{baek:Nature:2016}
	S.~K. Baek, H.-C. Jeong, C.~Hilbe, and M.~A. Nowak.
	\newblock Comparing reactive and memory-one strategies of direct reciprocity.
	\newblock \emph{Scientific Reports}, 6\penalty0 (1):\penalty0 25676--25676,
	2016.
	\newblock \doi{10.1038/srep25676}.
	
	\bibitem[Hilbe et~al.(2017)Hilbe, Martinez-Vaquero, Chatterjee, and
	Nowak]{hilbe:PNAS:2017}
	C.~Hilbe, L.~A. Martinez-Vaquero, K.~Chatterjee, and M.~A. Nowak.
	\newblock Memory-n strategies of direct reciprocity.
	\newblock \emph{Proceedings of the National Academy of Sciences}, 114\penalty0
	(18):\penalty0 4715--4720, 2017.
	\newblock \doi{10.1073/pnas.1621239114}.
	
	\bibitem[Nowak and Sigmund(1993)]{nowak:WSLS}
	M.~Nowak and K.~Sigmund.
	\newblock {A strategy of win-stay, lose-shift that outperforms tit-for-tat in
		the Prisoner's Dilemma game}.
	\newblock \emph{Nature}, 364\penalty0 (6432):\penalty0 56--58, 1993.
	\newblock \doi{10.1038/364056a0}.
	
	\bibitem[Rand et~al.(2009)Rand, Ohtsuki, and Nowak]{rand:JTB:2009}
	D.~G. Rand, H.~Ohtsuki, and M.~A. Nowak.
	\newblock {Direct reciprocity with costly punishment: Generous tit-for-tat
		prevails}.
	\newblock \emph{Journal of Theoretical Biology}, 256\penalty0 (1):\penalty0
	45--57, 2009.
	\newblock \doi{10.1016/j.jtbi.2008.09.015}.
	
	\bibitem[Hilbe et~al.(2013)Hilbe, Nowak, and Traulsen]{hilbe:adaptive}
	C.~Hilbe, M.~A. Nowak, and A.~Traulsen.
	\newblock {Adaptive Dynamics of Extortion and Compliance}.
	\newblock \emph{PloS ONE}, 8\penalty0 (11):\penalty0 e77886--e77886, 2013.
	\newblock \doi{10.1371/journal.pone.0077886}.
	
	\bibitem[Stewart and Plotkin(2013)]{stewart:PNAS:2013}
	A.~J. Stewart and J.~B. Plotkin.
	\newblock {From extortion to generosity, evolution in the Iterated Prisoner's
		Dilemma}.
	\newblock \emph{Proceedings of the National Academy of Sciences}, 110\penalty0
	(38):\penalty0 15348--15353, 2013.
	\newblock \doi{10.1073/pnas.1306246110}.
	
	\bibitem[Hilbe et~al.(2015{\natexlab{a}})Hilbe, Traulsen, and
	Sigmund]{hilbe:GEB:2015}
	C.~Hilbe, A.~Traulsen, and K.~Sigmund.
	\newblock {Partners or rivals? Strategies for the iterated prisoner's dilemma}.
	\newblock \emph{Games and Economic Behavior}, 92:\penalty0 41--52,
	2015{\natexlab{a}}.
	\newblock \doi{10.1016/j.geb.2015.05.005}.
	
	\bibitem[Hao et~al.(2018)Hao, Li, and Zhou]{payoffcontrol}
	D.~Hao, K.~Li, and T.~Zhou.
	\newblock {Payoff Control in the Iterated Prisoner's Dilemma}.
	\newblock In \emph{Proceedings of the Twenty-Seventh International Joint
		Conference on Artificial Intelligence, {IJCAI-18}}, pages 296--302.
	International Joint Conferences on Artificial Intelligence Organization, 7
	2018.
	\newblock \doi{10.24963/ijcai.2018/41}.
	
	\bibitem[Hilbe et~al.(2018{\natexlab{a}})Hilbe, Chatterjee, and
	Nowak]{partnersrivals}
	C.~Hilbe, K.~Chatterjee, and M.~A. Nowak.
	\newblock Partners and rivals in direct reciprocity.
	\newblock \emph{Nature Human Behaviour}, 2\penalty0 (7):\penalty0 469--477,
	2018{\natexlab{a}}.
	\newblock \doi{10.1038/s41562-018-0320-9}.
	
	\bibitem[Hauert et~al.(2006)Hauert, Michor, Nowak, and
	Doebeli]{hauert:JTB:2006}
	C.~Hauert, F.~Michor, M.~A. Nowak, and M.~Doebeli.
	\newblock Synergy and discounting of cooperation in social dilemmas.
	\newblock \emph{Journal of Theoretical Biology}, 239\penalty0 (2):\penalty0
	195--202, 2006.
	\newblock \doi{10.1016/j.jtbi.2005.08.040}.
	
	\bibitem[Boerlijst et~al.(1997)Boerlijst, Nowak, and
	Sigmund]{boerlijst:AMM:1997}
	M.~C. Boerlijst, M.~A. Nowak, and K.~Sigmund.
	\newblock {Equal Pay for All Prisoners}.
	\newblock \emph{The American Mathematical Monthly}, 104\penalty0 (4):\penalty0
	303, 1997.
	\newblock \doi{10.2307/2974578}.
	
	\bibitem[McAvoy and Nowak(2019)]{mcavoy}
	A.~McAvoy and M.~A. Nowak.
	\newblock Reactive learning strategies for iterated games.
	\newblock \emph{Proceedings of the Royal Society A: Mathematical, Physical and
		Engineering Sciences}, 475\penalty0 (2223):\penalty0 20180819, 2019.
	\newblock \doi{10.1098/rspa.2018.0819}.
	
	\bibitem[Curry(1944)]{haskell:QAM:1944}
	H.~B. Curry.
	\newblock {The Method of Steepest Descent for Non-linear Minimization
		Problems}.
	\newblock \emph{Quarterly of Applied Mathematics}, 2\penalty0 (3):\penalty0
	258--261, 1944.
	
	\bibitem[Calamai and Mor\'{e}(1987)]{calamai:MP:1987}
	P.~H. Calamai and J.~J. Mor\'{e}.
	\newblock Projected gradient methods for linearly constrained problems.
	\newblock \emph{Mathematical Programming}, 39\penalty0 (1):\penalty0 93--116,
	1987.
	\newblock \doi{10.1007/BF02592073}.
	
	\bibitem[Hilbe et~al.(2015{\natexlab{b}})Hilbe, Wu, Traulsen, and
	Nowak]{hilbe:JTB:2015}
	C.~Hilbe, B.~Wu, A.~Traulsen, and M.~A. Nowak.
	\newblock Evolutionary performance of zero-determinant strategies in
	multiplayer games.
	\newblock \emph{Journal of Theoretical Biology}, 374:\penalty0 115--124,
	2015{\natexlab{b}}.
	\newblock \doi{10.1016/j.jtbi.2015.03.032}.
	
	\bibitem[Shapley(1953)]{shapley:PNAS:1953}
	L.~S. Shapley.
	\newblock {Stochastic Games}.
	\newblock \emph{Proceedings of the National Academy of Sciences}, 39\penalty0
	(10):\penalty0 1095--1100, 1953.
	\newblock \doi{10.1073/pnas.39.10.1095}.
	
	\bibitem[Hilbe et~al.(2018{\natexlab{b}})Hilbe, {\v{S}}imsa, Chatterjee, and
	Nowak]{hilbe:Nature:2018}
	C.~Hilbe, {\v{S}}.~{\v{S}}imsa, K.~Chatterjee, and M.~A. Nowak.
	\newblock Evolution of cooperation in stochastic games.
	\newblock \emph{Nature}, 559\penalty0 (7713):\penalty0 246--249,
	2018{\natexlab{b}}.
	\newblock \doi{10.1038/s41586-018-0277-x}.
	
\end{thebibliography}
\end{document}